\documentclass[11pt]{article}

\usepackage[utf8]{inputenc}
\usepackage{authblk}

\usepackage{amsfonts,amsmath,amssymb,mathtools,amsthm}
\usepackage{graphicx}
\usepackage{stmaryrd}
\usepackage{thmtools}
\DeclareGraphicsExtensions{.eps}
\usepackage{float}

\usepackage{xcolor}
\definecolor{linkcolour}{rgb}{0.15,0.15,0.55}
\definecolor{urlcolour}{rgb}{0.15,0.15,0.55}
\definecolor{citecolour}{rgb}{0.15,0.15,0.55}

\usepackage[linktoc=all]{hyperref}
	\hypersetup{
		colorlinks = true,
		linkcolor = linkcolour,
		urlcolor = citecolour,
		citecolor = citecolour,
		linktoc = all,
		hypertexnames = false,
		unicode = true,
		bookmarksnumbered = false,
		pdfmenubar = true,
		pdftoolbar = true}
\hypersetup{
	pdfauthor 
		= {Mohamad Alameddine, Olivier Marchal, Nikita Nikolaev, Nicolas Orantin},
	pdftitle 
		= {Exact WKB solutions and direct monodromy problem associated to $0$-parameter solutions of the Painlev\'{e} I equation},
	pdfsubject
		= {},
	pdfcreator
		= {},
	pdfproducer
		= {},
	pdfkeywords
		= {},
}
\usepackage
	[top=2cm,
	bottom=2cm,
	left=2cm,
	right=2cm,
	a4paper]
	{geometry}
	
	\usepackage[margin=2cm]{caption}

\newcommand\encadremath[1]{\vbox{\hrule\hbox{\vrule\kern8pt
\vbox{\kern8pt \hbox{$\displaystyle #1$}\kern8pt}
\kern8pt\vrule}\hrule}}
\def\enca#1{\vbox{\hrule\hbox{
\vrule\kern8pt\vbox{\kern8pt \hbox{$\displaystyle #1$}
\kern8pt} \kern8pt\vrule}\hrule}}

\newcommand\framefig[1]{
\begin{figure}[bth]
\hrule\hbox{\vrule\kern8pt
\vbox{\kern8pt \vbox{
\begin{center}
{#1}
\end{center}
}\kern8pt}
\kern8pt\vrule}\hrule
\end{figure}
}

\newcommand\figureframex[3]{
\begin{figure}[bth]
\hrule\hbox{\vrule\kern8pt
\vbox{\kern8pt \vbox{
\begin{center}
{\mbox{\epsfxsize=#1.truecm\epsfbox{#2}}}
\end{center}
\caption{#3}
}\kern8pt}
\kern8pt\vrule}\hrule
\end{figure}
}
\newcommand\figureframey[3]{
\begin{figure}[bth]
\hrule\hbox{\vrule\kern8pt
\vbox{\kern8pt \vbox{
\begin{center}
{\mbox{\epsfysize=#1.truecm\epsfbox{#2}}}
\end{center}
\caption{#3}
}\kern8pt}
\kern8pt\vrule}\hrule
\end{figure}
}

\makeatletter
\@addtoreset{equation}{section}
\makeatother
\newtheorem{theorem}{Theorem}[section]

\newtheorem{proposition}{Proposition}[section]
\newtheorem{lemma}{Lemma}[section]
\newtheorem{corollary}{Corollary}[section]

\theoremstyle{definition}
\newtheorem{remark}{Remark}[section]
\newtheorem{definition}{Definition}[section]

\def\br{\begin{remark}\rm\small}
\def\er{\end{remark}}
\def\bt{\begin{theorem}}
\def\et{\end{theorem}}
\def\bd{\begin{definition}}
\def\ed{\end{definition}}
\def\bp{\begin{proposition}}
\def\ep{\end{proposition}}
\def\bl{\begin{lemma}}
\def\el{\end{lemma}}
\def\bc{\begin{corollary}}
\def\ec{\end{corollary}}
\def\beaq{\begin{eqnarray}}
\def\eeaq{\end{eqnarray}}

\theoremstyle{definition}

\newcommand{\be}{\begin{equation}}
\newcommand{\ee}{\end{equation}}
\newcommand{\beq}{\begin{equation}}
\newcommand{\eeq}{\end{equation}}
\newcommand{\bea}{\begin{eqnarray}}
\newcommand{\eea}{\end{eqnarray}}
\newcommand{\beqq}{\begin{equation*}}
\newcommand{\eeqq}{\end{equation*}}
\newcommand{\beaa}{\begin{eqnarray*}}
\newcommand{\eeaa}{\end{eqnarray*}}

\newcommand{\Tr}{{\operatorname {Tr}}}

\newcommand{\diag}{{\operatorname{diag}}}

\newcommand{\td}{\tilde}

\newcommand{\commutator}[2]{[#1,#2]}

\usepackage{subcaption}
\usepackage{changepage}

\newcommand{\Res}{\mathop{\,\rm Res\,}}

\begin{document}

\title{A Twisted $\mathfrak{sl}_2(\mathbb{C})$ Isomonodromic-Isospectral Correspondence}

\author{Mohamad Alameddine}
\vspace{5cm}

\affil{\small
Universit\'{e} Jean Monnet, Institut Camille Jordan 
\footnote{Institut Camille Jordan, CNRS UMR 5208,
Les Forges 2, 20 Rue du Dr Annino, 42000 Saint-\'{E}tienne, France. \\
mohamad.alameddine@univ-st-etienne.fr }}

\date{}
\maketitle

\abstract{The aim of this article is to generalize the isomonodromic-isospectral correspondence for meromorphic connections of rank $2$ over $\mathbb{P}^1$ to the twisted case. More specifically, the construction of the isospectral approach is provided for the Painlev\'{e}  I hierarchy, then two maps are constructed, one linking the sets of isomonodromic and isospectral Hamiltonians, another linking the set of apparent singularities to a set of isospectral coordinates.  }

\tableofcontents

 \newpage
\section{Introduction}
The study of meromorphic connections over vector bundles of Riemann surfaces presents a vast subject of modern mathematics, it is a playground for many different ideas coming from different fields (algebraic geometry, differential geometry, mathematical physics). Deformations preserving the monodromy of these connections, an idea first introduced by Riemann, links the subject to symplectic and Hamiltonian geometry. Nowadays, Fuchsian systems are relatively well understood and the reader is referred to Chap. $3$ of \cite{FromGaussToPainleve} for a review on the link with the Painlev\'{e} equations. These equations appear as compatibility conditions upon considering isomonodromic deformations, a fact first observed by R. Fuchs on what is now known as the Fuchsian system leading to the Painlev\'{e}  VI equation. However, the case of general isomonodromic deformations of linear ordinary differential equations with irregular singularities is much more complicated than the case of Fuchsian singularities and is still actively studied. For instance, seminal contributions to the theory were made by the Japanese school of M. Jimbo, T. Miwa, and K. Ueno \cite{JimboMiwaUeno,JimboMiwa} in the case of generic (in the sense that the leading term of the connections at each pole has distinct eigenvalues) singularity structures with arbitrary order poles on arbitrary rank bundles. The vast family of nonlinear differential equations resulting from the isomonodromic deformations is the largest known to have the ``Painlev\'{e} property". The Painlev\'{e} equations that have been studied extensively appear in the irregular case, these equations are the result of a classification problem, and are required to preserve a set of generalized monodromy data, which is obtained by considering Stokes data around the irregular singularities in addition to the usual monodromies, this problem gained huge interest due to the irregular Riemann-Hilbert correspondence linking it to the study of character varieties of surfaces. The space of deformations itself needs to be extended beyond the space of complex structures of the Riemann surface by including the irregular type of the monodromy considered \cite{Boalch2001,Boalch2014}. 

\medskip

The Painlev\'{e} I hierarchy, whose first element is the Painlev\'{e} I equation, appears as an example of the twisted construction. In this work, the term twisted is employed to describe a twisted formal normal form of the associated meromorphic differential system realized upon choosing a trivialization for the connection (the ramified case). This is also manifested in the fact that the singular part of the connection matrix is non-diagonalizable corresponding to a ramified pole or a branching point of the associated spectral curve. In rank $2$, the twist is manifested by a nilpotent leading order. Other types of twists exist in the literature, they are used for several purposes such as the classification of $2d$ topological quantum field theories and the counting of BPS states \cite{Dubrovin1996}. Despite this restriction, the case considered remains linked to the twisted character variety introduced in \cite{boalch2015twistedwildcharactervarieties} via the irregular Riemann-Hilbert correspondence. Additional motivation for considering twisted connections is the appearance of the linear meromorphic twisted differential system in the studies of $(3,p)$ string equations and the KP hierarchy which corresponds to a ``maximal twist'' in the rank $d$ connection. 

\medskip

 The strategy of the generic case of \cite{MarchalOrantinAlameddine2022} generalized later to the twisted case in \cite{Marchal2024} was based on passing to the oper form of the connection which is equivalent to the quantum curve satisfied by the first line of the matrix of horizontal sections, then, solve the compatibility equation in this particular gauge and derive the general Hamiltonian evolution of the apparent singularities. For the twisted case, one obtains non polynomial expressions for the Hamiltonian, this problem is solved by the introduction of a new set of symmetric coordinates similar to those introduced for the Painlev\'{e} II hierarchy in \cite{Mazzocco_2007}. Yamakawa's approach to isomonodromy is closely related, with different final objectives \cite{Yamakawa2017TauFA,Yamakawa2019FundamentalTwoForms}, he defines the isospectral Hamiltonians using the standard intrinsic residue-trace formulas at each pole and then imposes some conditions on the exterior derivatives relatively to irregular times to implicitly select Darboux coordinates for which the isospectral Hamiltonians match the isomonodromic ones. This strategy has also been used in \cite{Darboux_coord93} and is very efficient to obtain global properties of the symplectic Ehresmann connection such as flatness or completeness. However, it does not provide explicit coordinates so that one could not obtain a trivialization of the symplectic bundle (except from simple examples). Nevertheless, our final conclusion, in the twisted or generic settings, is that solving directly the isomonodromic compatibility equations or solving the isospectral condition are equivalent problems that share the same level of difficulty for $\mathfrak{sl_2}(\mathbb{C})$. 

\medskip

There are many ways to approach the problem, one of which is a viewpoint introduced by J. Harnad \cite{Harnad:1993hw} and J. Hurtubise \cite{HarnadHurtubise} built on moment map embeddings into central extensions of loop algebras. This strategy was developed to tackle isomonodromic deformations by considering first isospectral deformations (i.e. deformations that preserve the spectrum of the connection matrix) and then choose some specific set of coordinates so that an additional ``isospectral condition''
\beq \label{Condition} \delta_{\mathbf{t}} \td{L}(\lambda)=\partial_\lambda \td{A}(\lambda)\eeq
is realized \cite{BertolaHarnadHurtubise2022,Yamakawa2017TauFA}. The main feature of the additional isospectral condition is manifested by the identification of the Hamiltonian relatively to the chosen set of coordinates with the spectral Hamiltonians of the connection matrix (i.e. the coefficients of the expansion of the eigenvalues of the Lax matrix at each pole). This isospectral strategy has been used successfully for many cases such as Fuchsian singularities and all Painlev\'{e} cases and the general theory has been set up recently in \cite{BertolaHarnadHurtubise2022} for the generic case. Note that for Fuchsian systems, this condition turns out to be easily solved and the identification between isomonodromic and isospectral Hamiltonians is rather straightforward, yet this is evidently not the case when irregular singularities are considered. This was left as an open problem in \cite{BertolaHarnadHurtubise2022}, and then solved for the generic case of rank $2$ connections in \cite{Marchal_2024}, and as discussed in \cite{BertolaHarnadHurtubise2022} and also \cite{Marchal_2024}, finding the set of isospectral coordinates is far from trivial. This conclusion was established also in the work of Dubrovin and Mazzocco \cite{Dubrovin2007} (Ch. 5) where an explicit map was given relating the apparent singularities to isospectral coordinates for Schlesinger systems.

On the other side, the theory of isomonodromic deformations in the generic case developed in \cite{MarchalOrantinAlameddine2022}, has been extended to the twisted setup in \cite{Marchal2024} leading to the explicit Lax matrices and general Hamiltonian governing the Painlev\'{e} I hierarchy. Thus a natural open question is the following: can one extend the construction of the isospectral Hamiltonians to the twisted case and then relate the apparent singularities to a set of coordinates satisfying the isospectral condition ? Answering this question is the main objective of the present work and the main result is captured in the following theorem  

\begin{theorem} \label{Principalresult} There exists an explicit time dependent non-symplectic one-to-one map $f: (\mathbf{q}, \mathbf{p}) \rightarrow (\mathbf{u}, \mathbf{v})$ from the set of Darboux coordinates made from the apparent singularities $\mathbf{q}$ and their dual coordinates $\mathbf{p}$ of the oper form of the connection matrix to a set of isospectral coordinates $(\mathbf{u},\mathbf{v})$. This map can be seen as a change of coordinates between both approaches. Furthermore, this map is constructed explicitly and turns out to be of the same type for generic and twisted connections (Cf \cite{Marchal_2024} for the generic case).  
\end{theorem}

This result completes the answer for the rank $2$ side of the problem and many conclusions can be extracted from the explicit map relating the two sets of coordinates. First of all, the explicit map in both generic and twisted setups is a non-symplectic time dependent map. For the case $g=1$ (covering all the Painlev\'{e} equations), the map  is relatively simple, this is not the case for the higher genus setups where the map becomes highly non-trivial (higher elements of the Painlev\'{e} hierarchies). This explains why on a specific low genus case, the two approaches have huge similarities. Furthermore, from a geometric point of view,  when one equips the bundle over the base of times with an Ehresmann connection, the connection equipped with the apparent singularities and the one equipped with the isospectral coordinates differ (referred to as the Birkhoff connection in \cite{BertolaHarnadHurtubise2022}) and this difference is encoded in the differential one form $\omega_{JMU}$, this is a manifestation of the change of the symplectic structure for both sides. At last, the isospectral condition \eqref{Condition} transforms a derivative operator (on both twisted and generic sides) to a lower triangular Toeplitz invertible matrix, thus transforming the problem from a matrix algebra to a commuting subalgebra, admitting a recursive solution (the first one being trivial for the first coordinate). Finally let us mention that the application of our results to the Painlev\'{e} I hierarchy recovers the first members of the hierarchy, the Airy case ($g=0$), the Painlev\'{e} I equation ($g=1$) and the higher elements studied in \cite{Marchal2024}. 

\medskip

The paper is organized as follows, in \autoref{Sec2} the isospectral approach of \cite{BertolaHarnadHurtubise2022} is extended to the twisted case. Despite the fact that this is done only for the rank $2$, the strategy is indeed generalizable to higher rank cases. This remains beyond the scope of the present paper and is left for future works. In \autoref{Geo}, the goal is to review the twisted geometric construction of \cite{Marchal2024} leading to the general Hamiltonian giving the Painlev\'{e} I hierarchy, while explaining and clarifying some of the construction and providing an additional result concerning the reduction of the fundamental extended symplectic $2$-form
\begin{align}
 \Omega := \sum_{j=1}^g d q_j \wedge dp_j - \sum_{k=1}^{2 r_\infty-2} dt_{\infty,k} \wedge d\text{Ham}^{(\mathbf{e}_k)}    
\end{align}
explained in \autoref{Sym2form}. The last \autoref{Corrspondence} is devoted to the construction of the correspondence, in particular, two maps are constructed, one relating the set of isospectral Hamiltonians to the isomonodromic ones via the connection matrix $L(\lambda)$ thus allowing one to express the general Hamiltonian in terms of the isospectral Hamiltonians, and another relating the sets of coordinates to reach coordinates defined on the isospectral side.

\section{Twisted meromorphic connections over $\mathbb{P}^1$ }\label{Sec2}
\subsection{Twisted meromorphic connections and the Birkhoff factorization}
The starting point is the space of meromorphic connections admitting one fixed pole at infinity of arbitrary order over the projective line $\mathbb{P}^1$
\begin{definition}[Space of meromorphic connections with a pole at infinity]
Let $r_\infty \geq 3$ be a given integer. Define the space
\beq
F_{\infty, r_\infty}:= \left\{\hat{L}(\lambda) = \sum_{k=1}^{r_\infty-1} \hat{L}^{[\infty,k]} \lambda^{k-1} \,\,/\,\, \{\hat{L}^{[\infty,k]}\} \in \left(\mathfrak{gl}_2(\mathbb{C})\right)^{r_\infty-1}\right\}/{GL}_2 
\eeq
where ${GL}_2$ acts simultaneously by conjugation on all the coefficients $\{\hat{L}^{[\infty,k]}\}_{1\leq k\leq r_\infty-1}$. The meromorphic connection is defined on the vector bundle above the projective line as
\begin{align}
    \nabla = d- \hat{L} (\lambda)d\lambda
\end{align}
which, upon choosing a trivialization of the bundle becomes the system
\beq \label{diffsys} d \hat{\Psi}= \hat{L}(\lambda) \hat{\Psi}  d\lambda \,\, \Leftrightarrow \, \, \partial_\lambda \hat{\Psi}= \hat{L}(\lambda) \hat{\Psi}\eeq
where $\hat{\Psi}$ is by definition the $2 \times 2$ matrix of horizontal sections solution of the associated meromorphic differential system.
\end{definition}
This first definition states immediately that we are in the irregular case since the pole is assumed to have an order $r_\infty \geq 3$, yet does not contain any restriction on the leading order of the connection matrix. These restrictions on $\hat{L}(\lambda)$ are of particular interest since they manifest the different setups that one would consider. For instance, the \textbf{\textit{generic}} setup of the problem is obtained when considering a diagonalizable leading order with distinct eigenvalues, this setup was considered in \cite{BertolaHarnadHurtubise2022,MarchalOrantinAlameddine2022} and the link between the isospectral and isomonodromic symplectic structures was established in \cite{Marchal_2024} for the $\mathfrak{sl}_2(\mathbb{C})$ case. The main objective of the present work is a generalization of the results of \cite{BertolaHarnadHurtubise2022} to the \textbf{\textit{twisted}} setup on the isospectral side and to establish an explicit link with the $\mathfrak{sl}_2(\mathbb{C})$ symplectic structure of the isomonodromic side leading to the Painlev\'{e} I hierarchy as shown in \cite{Marchal2024}. More specifically, it is assumed that the leading order $\hat{L}^{[\infty,r_\infty-1]}$ is non-diagonalizable. In the literature, this case is also referred to as ``ramified'' at infinity. To this end, consider the following definition 

\begin{definition}[Set of twisted meromorphic connections at infinity] \label{Def} Let $r_\infty\geq 3$ be a given integer. The subset $\hat{F}_{\infty, r_\infty}$ of twisted connections is defined by
\small{\begin{align}
 \hat{F}_{\infty, r_\infty}=\bigg\{\hat{L}(\lambda) = \sum_{k=1}^{r_\infty-1} \hat{L}^{[\infty,k]} \lambda^{k-1} \,\,/\,\, \{\hat{L}^{[\infty,k]}\} \in \left(\mathfrak{gl}_2(\mathbb{C})\right)^{r_\infty-1} \nonumber \\ 
\text{ and } \hat{L}^{[\infty,r_\infty-1]} \text{ is not diagonalizable } \bigg\}/{GL}_2
\end{align}}
\normalsize{}
\end{definition} 
The space $ \hat{F}_{\infty, r_\infty}$ is the main subject of investigation from many point of views, it is worth mentioning that it is a Poisson manifold and its structure is inherited from a corresponding loop algebra through \textit{dual moment embeddings} \cite{HarnadHurtubise}. As a Poisson manifold, it has dimension 
\begin{align}
    \dim  \hat{F}_{\infty, r_\infty} = 3 r_\infty - 2
\end{align}
obtained after the conjugation action of the reductive group. The point of view adopted in this work is based on insights from the work of P. Boalch \cite{Boalch2001,boalch2015twistedwildcharactervarieties}, in particular, the local diagonalization of \autoref{Birkhoff} introduces a local fundamental form locally in the neighborhood of the pole controlled by the set of irregular times. Then, deforming the base of times while preserving the monodromies turns the space of connections itself into a Poisson fibre bundle over this base. Before introducing the base of times, it is important to note that the twisted setting is also characterized by the fractional powers of the local coordinate appearing in the local fundamental form. This is also manifested by the necessity of a double ramified cover over the ramified pole
\begin{align}
    z_\infty(\lambda) := \lambda^{-\frac{1}{2}}.  
\end{align}
Working always on a trivialization of the bundle, let us recall the classical description of formal solutions near an isolated singular point of first order linear differential systems:

\begin{proposition}[Birkhoff factorization] \label{Birkhoff} Consider the system (\ref{diffsys}) of first order ordinary differential equations over $\mathbb{P}^1$, then there exists a local gauge matrix (in the neighborhood of the pole)
\beq G_\infty(z)=G_{\infty,-1}z+G_{\infty,0}+\sum_{k=1}^{\infty} G_{\infty,k}z^{-k} \,\text{ with }\, G_{\infty,-1} \,\text{ of rank 1}
\eeq
with $z^{-1}G_\infty(z)$ locally holomorphic, and such that locally the horizontal section $\Psi_\infty(z)\overset{\text{def}}{=}G_\infty(z) \check{\Psi}$ is a formal fundamental solution, also known as a Turritin-Levelt fundamental form (or Birkhoff factorization) admits an expansion:
\begin{align}
    \Psi_\infty(z) = \Psi_{\infty}^{(\text{reg})}(z) e^{T(z)}, \qquad \Psi_{\infty}^{(\text{reg})}(z) = \left( I + \sum_{j\geq 1} \Psi_j^\infty z^j \right)
\end{align}
where the matrix $T(z)$ is a diagonal matrix given by 
\begin{align}
    T(z) = \diag\left(-\sum_{k=1}^{2r_\infty-2} \frac{t_{\infty,k}}{k z^k} + \frac{1}{2} \ln \frac{1}{z}, -\sum_{k=1}^{2r_\infty-2} (-1)^{k}\frac{t_{\infty,k}}{k z^k} + \frac{1}{2}\ln \frac{1}{z} \right). 
\end{align}
 This gauge action is manifested on the connection by $L_\infty= G_\infty \check{L} G_{\infty}^{-1}+ (\partial_\lambda G_\infty)G_\infty^{-1}$ which has a diagonal singular part at $\infty$:
 \begin{align}
 L_\infty(z)&=&\diag\left(-\frac{1}{2}\sum_{k=1}^{2r_\infty-2} (-1)^{k}t_{\infty,k} z^{-(k-2)} + \frac{z^2}{4}, -\frac{1}{2}\sum_{k=1}^{2r_\infty-2}(-1)^{k} t_{\infty,k} z^{-(k-2)} + \frac{z^2}{4}  \right) + O(1)   
 \end{align}
 In particular, locally, this diagonalization corresponds to 
 \begin{align}
     L_\infty(\lambda)=&\diag\left(-\frac{1}{2}\sum_{k=1}^{2r_\infty-2} t_{\infty,k} \lambda^{\frac{k}{2}-1} + \frac{1}{4\lambda}, -\frac{1}{2}\sum_{k=1}^{2r_\infty-2} (-1)^{k} t_{\infty,k} \lambda^{\frac{k}{2}-1} + \frac{1}{4\lambda}  \right) + O(1)
 \end{align}
 In particular, note that the local expansion admits a constant monodromy in the twisted rank $2$ setup, in contrary to the generic case. 
\end{proposition}
Before stating some remarks of the above definition, let us first define the base of times parameterizing the irregular type, naturally introduced due to the local diagonalization. 
\begin{definition}[Birkhoff invariants] The complex numbers $\left(t_{\infty,k}\right)_{1\leq k\leq 2r_\infty-2}$ define the ``irregular times'' at infinity denoted $\mathbf{t}=\{(t_{\infty,k})_{1\leq k\leq 2r_\infty-2}\}$, they parameterize the irregular type of the connection matrix.
\end{definition}
\begin{remark}
In general, one would consider arbitrary singular diagonal coefficients $t_{\infty^{(1)} , k},t_{\infty^{(2)} , k}$, as is done in \cite{MarchalOrantinAlameddine2022} in the generic case. However in the twisted rank $2$ case one has $\Tr \hat{L} = \Tr L_\infty + O(1)$ so that one must have $t_{\infty^{(1)} , k} = (-1)^k t_{\infty^{(2)} , k}$ (due to the fact that the l.h.s. is polynomial in $\lambda$). This fact is used immediately in the construction, and is specific to the case at hand. For instance, this does not appear on the generic side. Furthermore, the constant monodromy in this case is the manifestation of an infinite cyclic group of monodromy associated to $\mathbb{P}^1 \setminus \{\infty \}$. 
\end{remark}

\subsection{Spectral invariants and the link with the Birkhoff factorization}
The primary objective in this section is the expression of the irregular times of the local diagonalization as residues of the spectral invariants presented in our first \autoref{Specinv1}. To this end, the notion of the spectral curve is first introduced. 
\begin{definition}[The spectral curve] The spectral curve $\mathcal{S} $ attached to the connection matrix is the affine degree $2$ plane curve obtained through a compactification of the characteristic equation 
\begin{align} \label{spectral}
    \det \left( y I - \hat{L}(\lambda) \right) = 0
\end{align}
by adding a point at the infinite pole. This allows one to view the differential $1-$form $y d \lambda$ as a global meromorphic differential on the compactified curve admitting thus a local Laurent series near the irregular infinite pole representing it. 
\end{definition}
It is crucial to note the difference between the setups at this stage, in particular, the generic case of \cite{BertolaHarnadHurtubise2022} is based on the assumption that the leading order of the connection matrix admits a simple spectrum (distinct simple eigenvalues), and thus the eigenvalues all admit different Laurent series around the irregular pole (or irregular poles if one considers additional finite poles). The setup of the present work differs from the generic case as there is not any assumption put on the spectrum of the leading order, nevertheless, the local diagonalization places us in a rather similar problem. In particular, from the Birkhoff diagonalization, there exists eigenvalues solutions of the compactified spectral curve $\tilde{\mathcal{S}}$ (of degree $d$ in the general rank $d$ setup). These eigenvalues are characterized by the presence of fractional powers of the local coordinate which is the manifestation of the ramification. \\

This perspective is adopted to our case, first, in order to catch the local behavior around the pole, the matrix $1-$form $L(\lambda) d\lambda$ is used locally, note that one has
\begin{align}
    \frac{d\lambda}{d z} =-2 \lambda^{\frac{3}{2}}
\end{align}
The next step is to choose a suitably normalized invertible generalized matrix $M(z)$ for the connection matrix $L(z)$. In particular, this matrix admits a Laurent series in the neighborhood of the infinite pole given by 
\begin{align} \label{Modloc}
    M(\lambda) \overset{\lambda\to \infty}{=} M^\infty(z) = \left( I + \sum_{j\geq1} F^\infty_{j} z^j \right).
\end{align}
In particular, this allows one to write 
\begin{align}
    L(\lambda) M^\infty(z) = - \frac{1}{2 \lambda^{\frac{3}{2}}} M^\infty(z) J^\infty(z),
\end{align}
Where the matrix $J^\infty(z)$ is the matrix containing the spectral invariants attached to the connection matrix through its spectral curve. This perspective allows one to conclude the following first result of the paper. 

\begin{theorem}\label{Specinv1}[Twisted residue formula for the Birkhoff invariants] The spectral invariants solution of the spectral plane curve (\ref{spectral}) are given in terms of the Birkhoff invariants by the following expansion
\begin{align} \label{expansion}
y_1(\lambda)\overset{\lambda \to \infty}{=} & - \frac{1}{2} \sum_{k=1}^{2r_\infty-2} t_{\infty,k} \lambda^{\frac{k}{2}-1} +O\left(\lambda^{-1}\right) \cr
y_2(\lambda)\overset{\lambda \to \infty}{=} & -\frac{1}{2} \sum_{k=1}^{2r_\infty-2} (-1)^kt_{\infty,k} \lambda^{\frac{k}{2}-1} +O\left(\lambda^{-1}\right)
\end{align}
In other words, the spectral invariants are given through the residue formula 
\begin{align}
    t_{\infty,k} =  \, \Res_{\lambda\to \infty} \,\, \lambda^{-\frac{k}{2} -\frac{1}{2}} \,\, y_1(\lambda) \, d \lambda
\end{align}
\end{theorem}
\begin{proof}
    Note that in a local neighborhood, the local expansion of the eigenvector matrix is given in (\ref{Modloc}). On the other hand, the same is true for the local diagonalization of \autoref{Birkhoff}. In particular, the differential system in a local neighborhood is given by 
    \begin{align}
        \frac{d  \Psi_{\infty}^{(\text{reg})}}{d z_\infty} +  \Psi_{\infty}^{(\text{reg})} \frac{d T(z_\infty)}{d z_\infty} = L(\lambda) \frac{d\lambda}{d z_\infty}  \Psi_{\infty}^{(\text{reg})}
    \end{align}
    One immediately realizes that locally, the singular part on the l.h.s. is expressed in the term $\Psi_{\infty}^{(\text{reg})} \frac{d T(z_\infty)}{d z_\infty}$, restricting to the singular behavior, this means that both singular expansions coincide. The identification of the orders in the local coordinate of 
    \begin{align}
        L  = \sum_{k=1}^{2 r_\infty-2} L_k \lambda^{\frac{k}{2} -1} + O(\lambda^{-1})
    \end{align}
    with the orders of the derivative of the irregular type $T(z_\infty)$ yields the result. In other words, the matrix $\Psi_{\infty}^{(\text{reg})}$ coincides under a suitable normalization with the matrix of eigenvectors in the neighborhood of the pole up to the specific order. This in particular provides the expansion in our case expressed in (\ref{expansion}).
\end{proof}
It is worth mentioning at this stage, that despite the strategy being similar to the generic setup, the twisted case has always additional technical subtleties that one needs to handle, for instance, the ramification of the pole and the double cover over $\mathbb{P}^1$ play an important role as they change the orders of the local expansion of the connection matrix in order to obtain the result which would be just not possible without the double cover. \\

There are several consequences of this correspondence between the isospectral invariants and the Birkhoff invariants, one very particular consequence is the definition of the isospectral Hamiltonians linked to the eigenvalues of the connection matrices. Note that our result is compatible with the explicit geometric construction of the next \autoref{Geo} where explicit asymptotics for the connection matrix are given. The next goal, is to establish the link between the isospectral Hamiltonians and the expansion of the spectral invariants, the main objective of the next \autoref{IsoHam}.

\subsection{Isospectral Hamiltonians}\label{IsoHam}
The goal of this section is the introduction of the \textbf{\textit{isospectral Hamiltonians}} in relation with the spectral invariants defined in the previous section. On the generic side, the starting point is the definition of the total differential form $\omega_{JMU}$ introduced by the Japanese school of Jimbo, Miwa and Ueno in \cite{JimboMiwaUeno,JimboMiwa}. It was proven in op. cit. that this differential is closed when restricted to the solution manifold of the isomonodromic equation and hence is locally exact (proved only in the generic setup). The generalization of these results to the twisted setup is still missing in the literature, it is well-known that to define this differential in general settings one needs to consider additional contributions. This is important since through this notion, one gets to the definition of the $\tau-$function seen as a section of a line bundle over the solution space of the isomonodromic equations. Several attempts and strategies to generalize this differential are present and in \cite{Bertola2005} the expansion of the $\tau-$function has been considered, the authors than show that $d  \ln{\tau} $ is closed, and therefore, is the differential of a locally well-defined function. Other extensions are present in the literature inspired from works of \cite{CIF_1982__17__A2_0}, this was used in \cite{10.1215/00127094-2017-0055} to solve the evaluation of the connection formulae for isomonodromic $\tau-$functions in an explicit way. 

\medskip

Another way of dealing with the present problem is to change the gauge into one where the connection matrix is diagonalizable but admitting an additional twisted Fuchsian pole \cite{hayford2025isingmodelcoupled2d}, this again takes us out of the definition of the Japanese school. Nevertheless, it was shown in \cite{Lisovyy:2016qig} for Painlev\'{e} I, that the problem is surmounted by simply ignoring these additional Fuchsian contributions. For the purpose of the present article (the definition of the isospectral Hamiltonians), results of $\cite{Bertola2005}$ are used, however, the definition used must not be thought of as a generalization of the differential to the twisted side, rather, it is a sufficient one for our purposes. In particular we define 

\begin{definition}[Isospectral Hamiltonians]\label{specinv2} The set of isospectral Hamiltonians $\mathbf{I_{\infty,k}}$ is composed of elements defined through the expansion of the eigenvalues of the connection matrix, they are given by
\begin{align}
    \Res_{\lambda \to \infty} \left(y_1+ (-1)^ky_2 \right) \frac{\lambda^{\frac{k}{2}}}{k} = I_{\infty,k} + (-1)^k I_{\infty,k}
\end{align}
In other words, the expansion of the eigenvalues is given by
\begin{align}
        y_1(\lambda)\overset{\lambda \to \infty}{=} & -\frac{1}{2} \sum_{k=1}^{2r_\infty-2} t_{\infty,k} \lambda^{\frac{k}{2}-1} + \sum_{k=1}^{2r_\infty-2} k I_{\infty,k} \lambda^{-\frac{k}{2}-1} + O\left( \lambda^{-r_\infty-2} \right) \cr
        y_2(\lambda)\overset{\lambda \to \infty}{=} & -\frac{1}{2} \sum_{k=1}^{2r_\infty-2} (-1)^kt_{\infty,k} \lambda^{\frac{k}{2}-1}+ \sum_{k=1}^{2r_\infty-2}  (-1)^kk I_{\infty,k} \lambda^{-\frac{k}{2}-1} + O\left( \lambda^{-r_\infty-2} \right)
\end{align}
    
\end{definition}

Note that by definition, the isospectral Hamiltonians (on the generic side) are defined through the differential $\omega_{JMU}$ and thus determine the corresponding $\tau-$function of the system. This has not been generalized completely for the ramified case despite notable quests searching for a general definition. This is why in this article we do not tackle this question and rely on the eigenvalue expansion for the main definition.

\medskip

One natural idea suggests the definition of the isomonodromic oper coefficients in the same way (defined in the next section). However, this is relatively difficult since one gets additional factors dependent on the Birkhoff invariants of the local formal solution. This is one of the advantages present on the geometric side, the general theory is independent of the polar structure and handles the ramification using the same strategy as we will see in the next section.

\section{Geometric construction of isomonodromic deformations of the Painlev\'{e} I hierarchy} \label{Geo}
This section is devoted to the geometric construction of the Hamiltonian representation of the twisted rank $2$ connection upon considering deformations of the base of times introduced in \autoref{Birkhoff}. This has been the main result of \cite{Marchal2024}, the construction is reviewed and complemented with minor results/clarifications of the construction. The main new result of this section is the reduction of the fundamental extended symplectic $2-$form $\Omega$ presented in \autoref{Sym2form}. Note however, that the full statements and the proofs are omitted, the interested reader is referred to op. cit. for the full discussion of the results leading to the Painlev\'{e} I hierarchy. The starting point is \autoref{Birkhoff} defining the local diagonalization along with the gauge action of the fixed complex reductive group $GL_2(\mathbb{C})$. 

\subsection{Normalized representative and the choice of coordinates}
Let us first consider \autoref{Def} which includes the conjugation action splitting our construction over the orbits of the space. Using this action one may always fix a representative of each orbit admitting a leading coefficient $\hat{L}^{[r_\infty-1]}$ having a lower triangular form and with equal diagonal elements. Furthermore, the remaining action of the Cartan subgroup of diagonal matrices allows one to fix the diagonal part of the subleading order $\hat{L}^{[r_\infty-2]}$ to have the same entries. On the other hand, the trace and the determinant of the local diagonalization satisfy 
\begin{align}
    \Tr\, \hat{L} = \Tr\, L_\infty + O(1), \qquad \det L_\infty = \det \left( \hat{L} + G_\infty^{-1} \partial_\lambda G_\infty\right)
\end{align}
After fixing the irregular type of the connection parameterized by the set of Birkhoff invariants, this allows one to identify the space as 
\beq\label{NormalizationInfty}
\begin{array}{ll}
\hat{\mathcal{M}}_{\infty,r_\infty,\mathbf{t}} \sim &\Big\{  \tilde{L}(\lambda) = {\displaystyle \sum_{k=1}^{r_\infty-1}} \tilde{L}^{[\infty,k]} \lambda^{k-1}
\,/\, \tilde{L}^{[\infty,r_\infty-1]}=\begin{pmatrix} -\frac{1}{2}t_{\infty,2r_\infty-2} & 0\\ \frac{1}{4} (t_{\infty,2r_\infty-3})^2 & -\frac{1}{2}t_{\infty,2r_\infty-2}\end{pmatrix}\cr
&\text{ and } \tilde{L}^{[\infty,r_\infty-2]} =\begin{pmatrix}-\frac{1}{2}t_{\infty,2r_\infty-4}&1\\ \delta_\infty&-\frac{1}{2}t_{\infty,2r_\infty-4} \end{pmatrix}, \, \delta_\infty \in \mathbb{C}
\Big\}
\end{array} .
\eeq
We use the notation $\tilde{L} (\lambda)$ whenever we consider such a representative. This space is the phase space of deformations over the base of Birkhoff invariants, it is in particular a symplectic variety seen as the product of the orbits quotient by the conjugation action (completely exhausted at this stage). It has dimension $2 r_\infty -6 =2g$, where $g$ is the genus of the spectral curve attached to the connection matrix. 

\begin{remark}
    The results presented in this article extends to the limit $t_{\infty,2r_\infty-3}=0$ and $ t_{\infty,2r_\infty-2} \neq 0$, this limit in particular corresponds to the case where the leading order of the connection is diagonalizable admitting a non-simple spectrum, this is actually a direct example of a diagonalizable matrix with a non-simple spectrum. Note however, that even in this limit, the construction of the connection is still twisted due to the fractional powers in the irregular type of the connection matrix. Yet in this limiting case, although the results of the paper remain valid, one needs to adapt most of the formulation. For instance, one needs to truncate the matrix $M_\infty$ which becomes non-invertible. Yet, studying this limit is an instance of the geometric theory developed in this work and could be seen as a particular case of the formalism, the conclusion is that despite the fact that the leading order of the connection matrix is diagonalizable, the formal normal form (or the irregular type) still admits fractional powers and hence the case would be always twisted.
\end{remark}

The symplectic nature of the space $\hat{\mathcal{M}}_{\infty,r_\infty,\mathbf{t}}$ motivates the definition of a set of coordinates, the idea behind this definition dates back to Garnier and is widely used in the study of isomonodromy systems. 
\begin{definition}[Apparent singularities] In the respective gauge, considering a representative as above, the entry $\left[\tilde{L}(\lambda)\right]_{1,2}$ is a monic polynomial function of $\lambda$ of degree $r_\infty-3=g$. The set of coordinates $(q_i)_{1\leq i\leq g}$ is defined as the zeroes of $\left[\tilde{L}(\lambda)\right]_{1,2}$
\beq\label{Conditionqi}
\forall\, i\in \llbracket 1,g\rrbracket \, : \; \left[\tilde{L}(q_i)\right]_{1,2} = 0.
\eeq  
This set is complemented by the dual coordinates defined by 
\beq \label{Conditionpi}
\forall\, i\in \llbracket 1,g\rrbracket \, : \; p_i:=\left[\tilde{L}(q_i)\right]_{1,1}.
\eeq
and the two sets of coordinates define a set of Darboux coordinates on the symplectic phase space. 
\end{definition}
In particular, this definition ensures that for any $i \in \llbracket 1,g \rrbracket$, the pair $\left(q_i,p_i\right)$ defines by definition a point on the spectral curve defined by $\det(y I_2 -\td{L}(\lambda))=0$. In other words, one has
\beq 
\forall\, i\in \llbracket 1,g\rrbracket \, : \; \det(p_i-\tilde{L}(q_i)) = 0.
\eeq

\subsection{Oper gauge and isomonodromic deformations}
In this section, the oper gauge of the system is introduced, this change of basis of the bundle is essential for the study of the compatibility equation introduced when considering deformations of the base of times. In general, consider a rank $2$ system given by 
\beq  \partial_\lambda \td{\Psi}(\lambda)=\td{L}(\lambda)\td{\Psi}(\lambda).\eeq
One may perform a bundle automorphism seen as a change of basis in order to represent the system in a scalar form, notably, one defines 
\beq \label{GaugeGexpr}  \Psi(\lambda)=G(\lambda) \td{\Psi}(\lambda) \,\,\text{with}\,\, G(\lambda)=\begin{pmatrix} 1&0\\ \td{L}_{1,1}& \td{L}_{1,2}\end{pmatrix}\eeq
In particular the horizontal sections $\Psi$ are now solutions of the companion-like system
\beq \label{CompanionMatrix}\partial_\lambda \Psi(\lambda)=L(\lambda)\Psi(\lambda)\,\,\text{with}\,\, L(\lambda)=\begin{pmatrix}0&1\\ L_{2,1}&L_{2,2}\end{pmatrix}\eeq 
given by
\bea \label{LInTermsOfTdL} L_{2,1}&=&-\det \td{L}+\partial_\lambda\td{L}_{1,1}-\td{L}_{1,1}\frac{\partial_\lambda\td{L}_{1,2}}{\td{L}_{1,2}},\cr
L_{2,2}&=&\Tr\, \td{L} +\frac{\partial_\lambda\td{L}_{1,2}}{\td{L}_{1,2}}.
\eea
Note that the first line of $\Psi$ and $\td{\Psi}$ remains the same so that
\beq \Psi(\lambda)=\begin{pmatrix}\td{\Psi}_{1,1}(\lambda)& \td{\Psi}_{1,2}(\lambda)\\ \partial_\lambda \td{\Psi}_{1,1}(\lambda)& \partial_\lambda \td{\Psi}_{1,2}(\lambda) \end{pmatrix}=\begin{pmatrix}\psi_{1}(\lambda)& \psi_2(\lambda)\\ \partial_\lambda \psi_{1}(\lambda)& \partial_\lambda \psi_{2}(\lambda) \end{pmatrix}.\eeq
The companion-like system \eqref{CompanionMatrix} is equivalent to say that $\psi_1$ and $\psi_2$ satisfy the linear ODE:
\beq \left(\left[\partial_{\lambda}\right]^2 -L_{2,2}(\lambda)\partial_\lambda -L_{2,1}(\lambda)\right)\psi_i=0\eeq 
which is sometimes referred to as the ``quantum curve''. In order to relate both sides (isomonodromic and isospectral), explicit expressions are needed and are recorded in the following proposition. 

\begin{proposition}[Explicit formulation of \cite{Marchal2024}]
The gauge matrix leading to the oper system is defined explicitly as
\begin{align}
  G(\lambda) :=  \begin{pmatrix} 1&0\\ -Q(\lambda)-\left(\frac{1}{2}t_{\infty,2r_\infty-2}\lambda+g_0 \right) \underset{j=1}{\overset{g}{\prod}}(\lambda-q_j)&\underset{j=1}{\overset{g}{\prod}}(\lambda-q_j)\end{pmatrix}
\end{align}
where the values of $Q(\lambda)$ (the Lagrange polynomial) and $g_0$ are given by 
\begin{align}
\label{DefQ2}  Q(\lambda,\hbar) =& -\sum_{i=1}^g p_i \prod_{j\neq i}\frac{\lambda-q_j}{q_i-q_j} \nonumber \cr
 g_0=& \frac{1}{2}t_{\infty,2r_\infty-4}+\frac{1}{2}t_{\infty,r_\infty-2}\underset{j=1}{\overset{g}{\sum}} q_j \nonumber
\end{align}
In particular, this gauge transformation leaves us with a system admitting a connection matrix given by 
\beq L(\lambda,\hbar)=\begin{pmatrix} 0&1\\ L_{2,1}(\lambda)&L_{2,2}(\lambda)\end{pmatrix}\eeq
with entries
\bea L_{2,2}(\lambda)&=&\td{P}_1(\lambda)+ \sum_{j=1}^g \frac{\hbar}{\lambda-q_j}\cr
L_{2,1}(\lambda)&=&-\td{P}_2(\lambda) +\sum_{k=0}^{r_\infty-4} H_{\infty,k}\lambda^k - \sum_{j=1}^g\frac{ p_j}{\lambda-q_j}
\eea
Coefficients $\left(H_{\infty,k}\right)_{0\leq k\leq r_\infty-4}$ are referred to as the \textit{\textbf{isomonodromic oper coefficients}} due to the fact that the reduced Hamiltonian is completely determined by them. The terms $\td{P}_1(\lambda)$ and $\td{P}_2(\lambda)$ are polynomials given by:
\begin{align}
 \label{DefP1}\td{P}_1(\lambda)=&\sum_{j=0}^{r_\infty-2} \td{P}_{\infty,j}^{(1)}\lambda^j=-\sum_{j=0}^{r_\infty-2}t_{\infty,2j+2}\lambda^j, \cr
\td{P}_2(\lambda)&=\sum_{k=r_\infty-3}^{2r_\infty-4} \td{P}_{\infty,k}^{(2)}\lambda^{k}.
\end{align}
where, the coefficients of the second polynomial are given by 
\begin{align}
\td{P}_{\infty,k}^{(2)}=&\frac{1}{4}\sum_{j=2k-2r_\infty+6}^{2r_\infty-2}(-1)^j t_{\infty,j}t_{\infty,2k-j+4} \,\,,\,\, \forall\,k\in \llbracket r_\infty-2, 2r_\infty-4\rrbracket \cr
 \td{P}_{\infty,r_\infty-3}^{(2)}=&\frac{1}{4}\sum_{j=1}^{2r_\infty-3}(-1)^j t_{\infty,j}t_{\infty,2r_\infty-j-2}
\end{align}
We also define
\beq \hat{P}_2(\lambda)=\td{P}_2(\lambda)- \sum_{k=0}^{r_\infty-4} H_{\infty,k}\lambda^{k}\eeq
where the $g=r_\infty-3$ coefficients $\left(H_{\infty,k}\right)_{0\leq k\leq r_\infty-4}$ remaining unknown at this stage are the isomonodromic oper coefficients whose expression is given later.
\end{proposition}
Note that in the oper gauge one sees immediately the reason why this choice of coordinates is referred to as the apparent singularities, they appear as zeroes of the determinant of the horizontal sections $ \td{\Psi}$. This gauge is less natural from a geometric point of view since it adds an additional polar structure on the connection matrix, yet the main advantage of the oper gauge is the practical side since one could solve specifically in this gauge the zero curvature equation of the system. This is mainly due to the fact that all the information is stored in the last line of the connection matrix (in two entries instead of four). One of the main results of \cite{Marchal2024} is not just the explicit asymptotics near the ramified pole for the connection, but also the explicit Hamiltonian reformulation of the pair of coordinates leading to the Painlev\'{e} I hierarchy after the Hamiltonian reduction by the action of the projective linear group done through the action of the projective linear transformations on the connection. In particular, isomonodromic deformations are defined by a general deformation vector of the tangent space. 

\begin{definition}\label{DefGeneralDeformationsDefinition} Define the following general deformation operators.
\beq \label{GeneralDeformationsDefinition} \mathcal{L}_{\boldsymbol{\alpha}}= \sum_{k=1}^{2r_\infty-2} \alpha_{\infty,k} \partial_{t_{\infty,k}}\eeq
where the vector $\boldsymbol{\alpha}\in \mathbb{C}^{2r_{\infty}-2}=\mathbb{C}^{2g+4}$ is defined by
\beq \boldsymbol{\alpha}= \sum_{k=1}^{2r_\infty-2} \alpha_{\infty,k}\mathbf{e}_{k}.\eeq
\end{definition}
Deforming the base of times amounts to acting with this deformation vector on the horizontal sections, this provides an additional auxiliary matrix and the system is then given by 
\beq \label{LaxPairDef}
\left\{
\begin{array}{l}
\partial_\lambda \Psi (\lambda,\mathbf{t}) =L(\lambda) \Psi(\lambda,\mathbf{t})  \cr
\mathcal{L}_\alpha \Psi(\lambda,\mathbf{t}) =A_\alpha(\lambda) \Psi(\lambda,\mathbf{t})  \cr
\end{array}
\right.
\eeq
The important feature is that the auxiliary matrix admits a polar structure dominated by that of the connection matrix. Adding to that the asymptotics of the entries of the horizontal sections, one could extract an explicit expression for the entries of $A(\lambda)$. This then serves as the input of the zero curvature equation ensuring the compatibility of (\ref{LaxPairDef}), this equation is given by 
\beq
\partial_\lambda A_\alpha(\lambda) - \mathcal{L}_\alpha L(\lambda) + \left[L(\lambda) , A_\alpha(\lambda) \right] = 0.
\eeq
In particular, the dependence of the Lax pair $(L,A)$ on the set of coordinates enables one to extract the evolution of these coordinates relative to the deformation vector $\mathcal{L}_\alpha$. Before giving the explicit Hamiltonian structure attached to the deformations considered, it is worth mentioning that considering isomonodromic deformations and restricting the connection matrix and its auxiliary counter-part to the isomonodromic system of solutions of the system \eqref{LaxPairDef}, one ensures the flatness of the connection $1-$form 
\begin{align}
    \Upsilon := - L(\lambda) d \lambda - A_{\boldsymbol{\alpha}} (\lambda) d \mathbf{t} 
\end{align}
In other words one has 
\begin{align}
d \Upsilon + \commutator{\Upsilon}{\Upsilon} = 0    
\end{align}
due to the compatibility of the isomonodromic system, thus the name \textit{zero curvature equation}. The link with the isospectral approach requires also the explicit formulation of the auxiliary matrix given in the following proposition. 

\begin{proposition}[Auxiliary matrices Prop. 4.2 and 4.3 of \cite{Marchal2024}]  Entry $\left[A_{\boldsymbol{\alpha}}(\lambda)\right]_{1,2}$ is given by
\beq \label{ExpressionA12} \left[A_{\boldsymbol{\alpha}}(\lambda)\right]_{1,2}=\nu^{(\boldsymbol{\alpha})}_{\infty,-1}\lambda+\nu^{(\boldsymbol{\alpha})}_{\infty,0} + \sum_{j=1}^g \frac{\mu^{(\boldsymbol{\alpha})}_j}{\lambda-q_j}.\eeq

Moreover, coefficients $\left(\nu^{(\boldsymbol{\alpha})}_{\infty,k}\right)_{-1\leq k\leq r_\infty-3}$ are determined by 
\beq \label{RelationNuAlphaInfty} M_\infty\begin{pmatrix} \nu^{(\boldsymbol{\alpha})}_{\infty,-1}\\ \nu^{(\boldsymbol{\alpha})}_{\infty,0}\\  \vdots \\ \nu^{(\boldsymbol{\alpha})}_{\infty,r_\infty-3}\end{pmatrix}=\begin{pmatrix}\frac{2\alpha_{\infty,2r_\infty-3}}{(2r_\infty-3)}\\ \frac{2\alpha_{\infty,2r_\infty-5}}{(2r_\infty-5)}\\ \vdots \\ \frac{2\alpha_{\infty,1}}{1} \end{pmatrix}\eeq
where $M_\infty$ is a lower triangular Toeplitz matrix of size $(r_\infty-1)\times(r_\infty-1)$ independent of the deformation $\boldsymbol{\alpha}$:
\beq\label{MatrixMInfty} M_\infty=\begin{pmatrix}t_{\infty,2r_\infty-3}&0&\dots& &0\\
t_{\infty,2r_\infty-5}& t_{\infty,2r_\infty-3}&0& \ddots& \vdots\\
\vdots& \ddots& \ddots&&\vdots \\
t_{\infty,3}& & \ddots &\ddots&0 \\
t_{\infty,1}& t_{\infty,3} & \dots & &t_{\infty,2r_\infty-3}\end{pmatrix}.
\eeq
Coefficients $\left(\mu^{(\boldsymbol{\alpha})}_j\right)_{1\leq j\leq g}$ are determined by the linear system
\beq \label{RelationNuMuMatrixForm} V_\infty\begin{pmatrix}\mu^{(\boldsymbol{\alpha})}_1\\ \vdots\\ \vdots\\\mu^{(\boldsymbol{\alpha})}_g\end{pmatrix}= \begin{pmatrix}\nu^{(\boldsymbol{\alpha})}_{\infty,1}\\ \nu^{(\boldsymbol{\alpha})}_{\infty,2}\\\vdots \\ \nu^{(\boldsymbol{\alpha})}_{\infty,r_\infty-3}\end{pmatrix}\eeq
where $V_\infty$ is the $(r_{\infty}-3)\times (r_\infty-3)$ Vandermonde matrix 
\beq \label{DefVinfty}V_\infty=\begin{pmatrix}1&1 &\dots &\dots &1\\
q_1& q_2&\dots &\dots& q_{g}\\
\vdots & & & & \vdots\\
\vdots & & & & \vdots\\
q_1^{r_\infty-4}& q_2^{r_\infty-4} &\dots & \dots& q_{g}^{r_\infty-4}\end{pmatrix}\,,\, 
\eeq
While $\left[A_{\boldsymbol{\alpha}}(\lambda)\right]_{1,1}$ is given by
\beq \label{ExpressionA11} \left[A_{\boldsymbol{\alpha}}(\lambda)\right]_{1,1}=\sum_{i=0}^{r_\infty-1}c^{(\boldsymbol{\alpha})}_{\infty,i}\lambda^i+\sum_{j=1}^g\frac{\rho^{(\boldsymbol{\alpha})}_j}{\lambda-q_j}.\eeq
with 
\beq \forall\, j\in \llbracket 1,n\rrbracket \,:\, \rho^{(\boldsymbol{\alpha})}_j=-\mu^{(\boldsymbol{\alpha})}_j p_j\eeq
Coefficients $\left(c^{(\boldsymbol{\alpha})}_{\infty,k}\right)_{1\leq k\leq r_\infty-1}$  are determined by
\beq \label{calphaexpr} M_\infty \begin{pmatrix} c^{(\boldsymbol{\alpha})}_{\infty,r_\infty-1}\\ \vdots\\c^{(\boldsymbol{\alpha})}_{\infty,k}  \\\vdots \\ c^{(\boldsymbol{\alpha})}_{\infty, 1}\end{pmatrix}=\begin{pmatrix}\frac{\alpha_{\infty,2r_\infty-3}}{2r_\infty-3}t_{\infty,2r_\infty-2}-\frac{\alpha_{\infty,2r_\infty-2}}{2r_\infty-2}t_{\infty,2r_\infty-3}\\ \vdots\\ \underset{m=k}{\overset{r_\infty-1}{\sum}}\left(\frac{\alpha_{\infty,2k+2r_\infty-2m-3}}{2k+2r_\infty-2m-3}t_{\infty,2m}-\frac{\alpha_{2k+2r_\infty-2m-2}}{2k+2r_\infty-2m-2}t_{\infty,2m-1}\right)  \\ \vdots\\ \underset{m=1}{\overset{r_\infty-1}{\sum}}\left(\frac{\alpha_{\infty,2r_\infty-2m-1}}{2r_\infty-2m-1}t_{\infty,2m}-\frac{\alpha_{2r_\infty-2m}}{2r_\infty-2m}t_{\infty,2m-1}\right) \end{pmatrix}
\eeq
\end{proposition}

In \cite{Marchal2024}, explicit expressions for these Hamiltonians are given which we record in the following theorem. 

\begin{theorem}[Hamiltonian evolution of the Darboux coordinates Theorem 5.1 of \cite{Marchal2024}] The evolution of the coordinates ensuring the compatibility of the system is Hamiltonian 
\beq \forall\, j\in \llbracket 1,g\rrbracket \,:\, \mathcal{L}_{\boldsymbol{\alpha}}[q_j]=\frac{\partial \text{Ham}^{(\boldsymbol{\alpha})}(\mathbf{q},\mathbf{p})}{\partial p_j}\, \text{ and }\, \mathcal{L}_{\boldsymbol{\alpha}}[p_j]=-\frac{\partial \text{Ham}^{(\boldsymbol{\alpha})}(\mathbf{q},\mathbf{p})}{\partial q_j}.\eeq
with the Hamiltonian of the system given by 
\beq \label{DefHam}\text{Ham}^{(\boldsymbol{\alpha})}(\mathbf{q},\mathbf{p}) =\sum_{k=0}^{r_\infty-4} \nu_{\infty,k+1}^{\boldsymbol{(\alpha)}}H_{\infty,k}-\hbar \sum_{j=1}^g\sum_{k=1}^{r_\infty-1}c^{(\boldsymbol{\alpha})}_{\infty,k}q_j^{k}-\hbar \nu^{(\boldsymbol{\alpha})}_{\infty,0}\sum_{j=1}^{g} p_j-\hbar \nu^{(\boldsymbol{\alpha})}_{\infty,-1}\sum_{j=1}^g q_jp_j,
\eeq
At this stage, the isomonodromic oper coefficients are explicitly determined via 
\beq \label{DefCi2}
(V_\infty)^{t}\begin{pmatrix}H_{\infty,0}\\ \vdots\\ H_{\infty,r_\infty-4} \end{pmatrix}=\begin{pmatrix} p_1^2- \td{P}_1(q_1)p_1 +\td{P}_2(q_1)+\hbar \underset{i\neq 1}{\sum}\frac{p_i-p_1}{q_1-q_i}\\
\vdots\\
p_g^2- \td{P}_1(q_g)p_g+\td{P}_2(q_g)+\hbar \underset{i\neq g}{\sum}\frac{p_i-p_g}{q_g-q_i}
\end{pmatrix}
\eeq
\end{theorem}
Let us draw some conclusions at the end of this section corresponding to the explicit construction of the twisted Hamiltonian. First of all, the general Hamiltonian is a linear combination of the isomonodromic oper coefficients. Note that the coefficients $\left(H_{\infty,k}\right)_{0\leq k\leq r_\infty-4}$ do not depend on the isomonodromic deformations and correspond to the unknown coefficients of $L_{2,1}(\lambda)$. Furthermore, additional terms appear as linear combinations of the coordinates. These terms and the additional fact that the dimension of the tangent space is higher than the genus $g$ of the spectral curve (or dimension of the symplectic manifold $\hat{\mathcal{M}}_{\infty,r_\infty,\mathbf{t}}$) motivates a Hamiltonian reduction in the twisted setting. This is the main objective of the next section in which the explicit twisted Hamiltonian is reduced to the Hamiltonian governing the Painlev\'{e} I hierarchy. 

\subsection{Reduction of the symplectic structure}
The symplectic reduction consists in tracking a group action (the projective linear group $PGL_2 (\mathbb{C})$) on the connection, it is well known that the reduction of the connection eliminating the trace (from $\mathfrak{gl}_2$ to $\mathfrak{sl}_2$) does not affect the symplectic structure upon taking deformations of the base of times. Furthermore, the symplectic structure is also an invariant of the action of the M\"{o}bius transformations. These are the main two ingredients that one needs to consider in the rank $2$ case to decompose the space of deformations of dimension $2g +4$ into two disjoint subspaces:
\begin{itemize}
    \item $\mathcal{T}_{trivial}$ of dimension $g+4$ consisting of deformation directions through which there is no evolution of the coordinates. 
    \item $\mathcal{T}_{iso}$ of dimension $g$ consisting of deformation directions along which the full symplectic structure is encoded. 
\end{itemize}
This decomposition allows one to obtain a general twisted Hamiltonian governing the elements of the Painlev\'{e} I hierarchy. In \cite{Marchal2024}, the authors gave an explicit detailed formulation of this decomposition reducing the Hamiltonian to a Liouville-integrable Hamiltonian system. In this article, an additional result is added to this formulation, the symplectic reduction of the \textit{\textbf{extended fundamental symplectic $2-$form}} attached to the symplectic structure \autoref{Sym2form}. This way, the reduction of the symplectic structure is complete. \\

The starting point of the reduction is the following change of basis of the tangent space to the base

\begin{definition}[Change of deformation basis, def 7.1 of \cite{Marchal2024}]\label{TrivialVectors} Define the following vectors of $\mathbb{C}^{2r_\infty-2}$ and their corresponding deformations.
\begin{align}
 \mathbf{w}_k=&\mathbf{e}_{2k} \,\,,\,\, \forall \, k\in \llbracket 1,r_\infty-1\rrbracket\cr
 \mathbf{u}_{k}=&\frac{1}{2}\sum_{m=1}^{r_\infty-k-2} (2m-1)t_{\infty,2m+1+2k}\mathbf{e}_{2m-1} +\frac{1}{2}\sum_{s=1}^{r_\infty-k-2}2s\, t_{\infty,2s+2k+2}\mathbf{e}_{2s}\cr
=&\frac{1}{2}\sum_{r=1}^{2r_\infty-2k-4} r\,t_{\infty,r+2k+2} \mathbf{e}_r  \,\,,\,\,  \forall \, k\in \llbracket -1,r_\infty-3\rrbracket
\end{align}
and denote:
\begin{align}
 \mathcal{U}_{\text{trivial}}=&\text{Span}\left\{\mathbf{w}_1,\dots,\mathbf{w}_{r_\infty-1},\mathbf{u}_{-1},\mathbf{u}_0\right\}\cr
\mathcal{U}_{\text{iso}}=&\text{Span}\left\{\mathbf{u}_1,\dots,\mathbf{u}_{r_\infty-3}\right\}
\end{align}
\end{definition}
In particular, one obtains the following evolutions relative to these vectors 
\begin{align}
 \mathcal{L}_{\mathbf{w}_k}[q_j]=&0 \,\, ,\,\, \forall\, k\in \llbracket 1,r_\infty-1\rrbracket \cr
\mathcal{L}_{\mathbf{w}_k}[p_j]=&-\frac{\hbar}{2} q_j^{k-1}\,\, ,\,\, \forall\, k\in \llbracket 1,r_\infty-1\rrbracket \cr
\mathcal{L}_{\mathbf{u}_{-1}}[q_j]=&-\hbar q_j\cr
\mathcal{L}_{\mathbf{u}_{-1}}[p_j]=&\hbar p_j\cr
\mathcal{L}_{\mathbf{u}_{0}}[q_j]=&-\hbar\cr
\mathcal{L}_{\mathbf{u}_{0}}[p_j]=&0
\end{align}
It turns out that despite the fact that some of these evolutions are not directly trivial, one still could define a symplectic ``shift'' of the set of coordinates in order to trivialize them completely. Before that, the decomposition of the tangent space motivates the choice of the $g+4$ trivial times components of $\mathcal{T}_{trivial}$ and complement them with $g$ non-trivial components of $\mathcal{T}_{iso}$.

\begin{definition}[Deformation times def 7.2 of \cite{Marchal2024}]\label{Times} The set of ``trivial times" is composed of :
\bea T_{\infty,k}&=&t_{\infty,2k} \,\, ,\,\, \forall\, k\in \llbracket 1,r_\infty-1\rrbracket\cr
T_2&=&\left(\frac{1}{2}t_{\infty,2r_\infty-3}\right)^{\frac{2}{2r_\infty-3}}\cr
T_1&=&\frac{t_{\infty,2r_\infty-5}}{2r_\infty-5} \left(\frac{1}{2}t_{\infty,2r_\infty-3}\right)^{-\frac{2r_\infty-5}{2r_\infty-3}}
\eea
Define also the $g=r_\infty-3$ ``isomonodromic'' times $\left(\tau_{k}\right)_{1\leq k\leq g}$,for all $k\in \llbracket 1, g\rrbracket$, by: 
\small{
\begin{align}
\tau_k=&\sum_{i=0}^{k-1}\frac{(-1)^i\left(\underset{s=1}{\overset{i}{\prod}} (2r_\infty-2k+2s-7)\right)  \left(\frac{1}{2}t_{\infty,2r_\infty-5}\right)^i \left(\frac{1}{2}t_{\infty,2r_\infty-3}\right)^{-\frac{(2r_\infty-3)i+2r_\infty-5-2k}{2r_\infty-3}} \frac{1}{2}t_{\infty, 2r_\infty-5-2k+2i} }{i!(2r_\infty-5)^i}\cr
&+ \frac{(-1)^{k}\left(\underset{s=1}{\overset{k}{\prod}}(2r_\infty-2k+2s-7)\right) \left(\frac{1}{2}t_{\infty, 2r_\infty-5}\right)^{k+1} \left(\frac{1}{2}t_{\infty,2r_\infty-3}\right)^{-\frac{(k+1)(2r_\infty-5)}{2r_\infty-3} }}{(k+1)(k-1)!(2r_\infty-5)^{k}} \nonumber
\end{align}}
\end{definition}
In particular, the definition of the trivial times ensures the following property 
\bea \mathcal{L}_{\mathbf{w}_k}[T_{\infty,j}]&=&\hbar\delta_{j,k}\,\,,\,\, \mathcal{L}_{\mathbf{u}_{-1}}[T_{\infty,j}]=\hbar j t_{\infty,2j}\,\,,\,\, \mathcal{L}_{\mathbf{u}_0}[T_{\infty,j}]=\hbar jt_{\infty,2j+2} \,\,,\,\,\forall\, j\in \llbracket 1, r_\infty-1\rrbracket\cr
\mathcal{L}_{\mathbf{w}_k}[T_2]&=&0 \,\,,\,\, \mathcal{L}_{\mathbf{u}_{-1}}[T_2]=\hbar T_2\,\,,\,\,\mathcal{L}_{\mathbf{u}_{0}}[T_2]=0\cr
\mathcal{L}_{\mathbf{w}_k}[T_1]&=&0 \,\,,\,\, \mathcal{L}_{\mathbf{u}_{-1}}[T_1]=0\,\,,\,\,\mathcal{L}_{\mathbf{u}_{0}}[T_1]=\hbar T_2
\eea
This in turn allows one to define the ``shifted'' coordinates in which one trivializes the evolutions with respect to the set of trivial deformations.

\begin{definition}\label{ShiftDarbouxCoordinates} The shifted coordinates $(\check{q}_j,\check{p}_j)_{1\leq j\leq g}$ are defined through the symplectic change of coordinates
\bea \check{q}_j&=&T_2 q_j+T_1\cr
\check{p}_j&=&T_2^{-1}\left(p_j-\frac{1}{2}\td{P}_1(q_j)\right) \,,\,\, \forall \, j\in \llbracket 1,g\rrbracket
\eea
In particular, using the decomposition of the tangent space, these coordinates satisfy 
\bea  \mathcal{L}_{\mathbf{w}_k}[\check{q}_j]&=&\mathcal{L}_{\mathbf{w}_k}[\check{p}_j]=0 \,\, ,\,\, \forall\, k\in \llbracket 1,r_\infty-1\rrbracket \cr
\mathcal{L}_{\mathbf{u}_{-1}}[\check{q}_j]&=&\mathcal{L}_{\mathbf{u}_{-1}}[\check{p}_j]=0\cr
\mathcal{L}_{\mathbf{u}_{0}}[\check{q}_j]&=&\mathcal{L}_{\mathbf{u}_{0}}[\check{p}_j]=0
\eea
Thus, justifying the term ``trivial'' deformations. 
\end{definition}
This definition enables one to perform a reduction of the symplectic structure manifested by the extended fundamental symplectic $2-$form, in particular, one has the following theorem.

\begin{theorem}[Reduction of the symplectic structure]\label{Sym2form} The \textit{\textbf{extended fundamental symplectic $2-$form}} characterizing the symplectic structure is defined by 
\begin{align}
 \Omega := \sum_{j=1}^g d q_j \wedge dp_j - \sum_{k=1}^{2 r_\infty-2} dt_{\infty,k} \wedge d\text{Ham}^{(\mathbf{e}_k)}    
\end{align}
In particular, the symplectic reduction is manifested by the reduction of $\Omega$ to 
\begin{align}
    \Omega=\sum_{j=1}^g d\check{q}_j \wedge d\check{p}_j- \sum_{\tau\in \mathcal{T}_{\text{iso}}}^g d\tau \wedge d \text{Ham}^{(\boldsymbol{\alpha}_{\tau})}.
\end{align}
where the vector $\boldsymbol{\alpha}_\tau$ is the vector corresponding to the tangent space of deformation $\partial_{\tau}$. 

\end{theorem}

\begin{proof}
    The proof follows from a direct computation using the decomposition of the tangent space along with the definition of trivial and isomonodromic times. It is detailed in \autoref{proof}.
\end{proof}

\begin{corollary}
    The symplectic structure manifested by the isomonodromic deformations of the horizontal sections is independent of the set of trivial times $\mathcal{T}_{trivial}$
\end{corollary}

The above result allows one to conclude that the symplectic structure is independent of the trivial times, and thus, one could fix these times to an arbitrary value without affecting the underlying symplectic structure. In \cite{Marchal2024}, a specific choice was made allowing one to reduce the Hamiltonian structure to a particular form, this form is equivalent to setting the even times $t_{\infty,2k} =0$, $t_{\infty,2r_\infty-3} =2$ and $t_{\infty,2r_\infty-5}=0$. This choice also is a manifestation of the traceless property of the connection matrix in the respective gauge. In terms of trivial times one takes

\begin{definition}[Canonical choice of the trivial times def 8.1 of \cite{Marchal2024}]\label{TrivialTimesChoice} Define the ``canonical choice of trivial times'' by choosing
\bea T_{\infty,k}&=&0 \,\,\,,\,\, \forall \, k\in \llbracket 0,r_\infty-1\rrbracket, \cr
T_{1}&=&0 ,\cr
T_{2}&=&1.
\eea 
\end{definition}
Let us note that this choice is not unique and one could evidently go for another choice of the trivial times without affecting the symplectic structure, the term ``canonical'' is employed since this particular choice allows one to get a particular reduction of the Hamiltonian. \\

In the rest of the paper, the trivial times are set to their canonical values. The canonical choice of trivial times implies that
\begin{itemize}\item All even irregular times are set to $0$: for all $k\in \llbracket 1, r_\infty-1\rrbracket$: $t_{\infty,2k}=0$.
\item $t_{\infty,2r_\infty-3}=2$ and $t_{\infty,2r_\infty-5}=0$.
\item $\td{P}_1$ is identically null. This is equivalent to say that $\td{L}$ and $\check{L}$ are traceless. This implies that under a potential additional trivial gauge transformation, one may choose a gauge in which $\td{L}$, $\check{L}$,$\td{A}_{\boldsymbol{\alpha}^{\tau}}$ and $\check{A}_{\boldsymbol{\alpha}^{\tau}}$ are traceless for any isomonodromic time $\tau\in \mathcal{T}_{\text{iso}}$.
\item The shifted Darboux coordinates are identical to the initial Darboux coordinates:
\beq \forall \, j\in \llbracket 1,g\rrbracket\,:\, \check{q}_j=q_j \,\text{ and }\, \check{p}_j=p_j\eeq
\item The isomonodromic times $\tau_k$ identify with an irregular time:
\beq \label{Identifisoirreg}\forall\, k\in \llbracket 1,g\rrbracket\,:\, \tau_k= \frac{1}{2}t_{\infty,2r_\infty-2k-5} \,\, \Leftrightarrow\,\, \frac{1}{2}t_{\infty,2k-1}=\tau_{r_\infty-k-2}\eeq
\item $\td{P}_2$ reduces to $\td{P}_2(\lambda)=-\lambda$ if $r_\infty=3$ or for $r_\infty\geq 4$:
\bea \label{ReducedtdP2}\td{P}_2(\lambda)&=&-\lambda^{2r_\infty-5}-\sum_{k=r_\infty-2}^{2r_\infty-7}\left(2\tau_{2r_\infty-k-6}+\sum_{m=k-r_\infty+6}^{r_\infty-3}\tau_{r_\infty-m-2}\tau_{r_\infty-k+m-5}\right)\lambda^k\cr
&&-\left(2\tau_{r_\infty-3}+\sum_{m=3}^{r_\infty-3}\tau_{r_\infty-m-2}\tau_{m-2}\right)\lambda^{r_\infty-3}
\eea
\end{itemize}

After the reduction one gets the following result 

\begin{theorem}[Hamiltonian representation for the canonical choice of trivial times Theorem 8.1 of \cite{Marchal2024}]\label{HamTheoremReduced} The canonical choice of the trivial times given by \autoref{TrivialTimesChoice} and the definition of trivial times (\autoref{Times}) imply that for any isomonodromic time $\tau\in \mathcal{T}_{\text{iso}}$:
\beq \label{DefHamReduced} \text{Ham}^{(\boldsymbol{\alpha}^\tau)}(\check{\mathbf{q}},\check{\mathbf{p}})=\sum_{k=0}^{r_\infty-4} \nu_{\infty,k+1}^{(\boldsymbol{\alpha}^\tau)}H_{\infty,k}\eeq
In other words, the Hamiltonian is a (time-dependent) linear combination of the isomonodromic oper coefficients $(H_{\infty,k})_{0\leq k\leq r_\infty-4}$ that are determined by
\beq  \label{ReducedDefCi2}\begin{pmatrix}1&\check{q}_1 &\dots &\dots &\check{q}_1^{g-1}\\
1& \check{q}_2&\dots &\dots&\check{q}_2^{g-1} \\
\vdots & & & & \vdots\\
\vdots & & & & \vdots\\
1& \check{q}_{g} &\dots & \dots& \check{q}_{g}^{g-1}\end{pmatrix}\begin{pmatrix} H_{\infty,0}\\ \vdots\\ \vdots\\ H_{\infty,r_\infty-4}\end{pmatrix}=
\begin{pmatrix} \check{p}_1^2+\td{P}_2(\check{q}_1)+\hbar \underset{i\neq 1}{\sum}\frac{\check{p}_i-\check{p}_1}{\check{q}_1-\check{q}_i}\\
\vdots\\ \vdots\\
\check{p}_g^2 +\td{P}_2(\check{q}_g)+\hbar \underset{i\neq g}{\sum}\frac{\check{p}_i-\check{p}_g}{\check{q}_g-\check{q}_i}
\end{pmatrix}
\eeq
where the coefficients of the linear combination are given by 
\beq\label{Matrixentries} \begin{pmatrix}1&0&\dots&\dots&\dots&0\\
0&1&0&\ddots&&0\\
\tau_1 &0&1&0&\ddots &\vdots\\
\vdots&\ddots &\ddots&\ddots&\ddots &\vdots\\
\vdots&\ddots &\ddots&\ddots&\ddots &\vdots\\
\tau_{g-2}&\tau_{g-3}&\dots& \tau_1&0&1 
\end{pmatrix}\begin{pmatrix}\nu_{\infty,1}^{(\boldsymbol{\alpha}^{\tau_j})}\\ \vdots\\ \vdots\\\nu_{\infty,r_\infty-3}^{(\boldsymbol{\alpha}^{\tau_j})}\end{pmatrix}=\frac{1}{2r_\infty-2j-5} \mathbf{e}_{j}
\eeq
\end{theorem}
Before concluding this section, let us note that the reduced Hamiltonian is valid for any choice of trivial times since it holds no dependence on these parameters, yet, different choices will indeed affect the choice of ``shifted'' coordinates. This Hamiltonian is the general Hamiltonian governing the Painlev\'{e} I hierarchy as one could vary the order of the pole and obtain the elements of the hierarchy. The purpose of the next section is thus to relate this Hamiltonian to the isospectral ones, thus obtaining a correspondence on both sides, note however that this does not ensure the correspondence between the symplectic structures as one needs to relate the coordinates used in both setups. This correspondence is the main objective of the next section.

\section{Twisted isomonodromic-isospectral correspondence}\label{Corrspondence}

The goal of this section is to establish a direct relation between the two sets of isomonodromic oper coefficients and isospectral Hamiltonians $\mathbf{H}_{\infty,k}$ and $\mathbf{I}_{\infty,k}$ defined in the previous two sections, thus relating the Hamiltonian formulation. Note that for the correspondence to be complete, one needs to build a map also between the set of coordinates whose evolution characterizes the symplectic structures on both sides. This is the ultimate goal of this section, and for this, it is reasonable to define what is a correspondence in this context: \\

\textit{\textbf{``A correspondence between the isospectral and isomonodromic symplectic structures is manifested by the existence of a map between the sets of Hamiltonians $\mathbf{I}_{\infty,k}$ and oper coefficients $\mathbf{H}_{\infty,k}$ along with a map relating two sets of coordinates defined on both sides.''}} \\

Following this, the section is split into two parts, each one of which treats a side of the correspondence while discussing relative consequences. The question of an isospectral coordinate system left open in \cite{BertolaHarnadHurtubise2022} is discussed in the second part for the twisted case.

\subsection{Relation between $\mathbf{H}_{\infty,k}$ and $\mathbf{I}_{\infty,k}$}
The first step is a relation between the sets $\mathbf{I}_{\infty,k}$ and $\mathbf{H}_{\infty,k}$, for this, note that one needs the initial gauge rather than the oper one, it is then straightforward to write 
\begin{align}
    \det \td{L} = & \det (L + G \partial_\lambda G^{-1}) = - [L]_{2,1} - [G \partial_\lambda G^{-1}]_{2,1} \cr
    = & \td{P}_2(\lambda) -\sum_{k=0}^{r_\infty-4} H_{\infty,k}\lambda^k + \sum_{j=1}^g \frac{p_j}{\lambda-q_j}+ [\td{L}]_{1,2} \partial_\lambda \frac{[\td{L}]_{1,1}}{[\td{L}]_{1,2}} 
\end{align}
The connection matrix $\td{L}$ in the initial gauge admits a polar structure bounded by the divisor considered containing only one ramified pole (fixed at $\{\infty \}$), henceforth, its determinant is determined by the singular part at the pole. The expansion of the determinant is completely determined once the behavior of the last term $[\td{L}]_{1,2} \frac{[\td{L}]_{1,1}}{[\td{L}]_{1,2}}$ is given. Note that the term depending on the coordinates holds no contribution to the expansion. On the other hand, from \autoref{specinv2}, one has
\begin{align}
    \det \td{L} =& - y_1 (\lambda) y_2( \lambda) =- \left( - \frac{1}{2} \sum_{k=1}^{2r_\infty-2} t_{\infty,k} \lambda^{\frac{k}{2}-1} + \sum_{k=1}^{2r_\infty-2} k I_{\infty,k} \lambda^{-\frac{k}{2}-1} + O\left( \lambda^{-r_\infty-2} \right) \right) \cr
    & \left( - \frac{1}{2} \sum_{k=1}^{2r_\infty-2} (-1)^k t_{\infty,k} \lambda^{\frac{k}{2}-1} + \sum_{k=1}^{2r_\infty-2} (-1)^k k I_{\infty,k}\lambda^{-\frac{k}{2}-1} + O\left( \lambda^{-r_\infty-2} \right) \right) \cr
    = & -\frac{1}{4} \sum_{k=1}^{2r_\infty-2} \sum_{j=1}^{2r_\infty-2} (-1)^j t_{\infty,k} t_{\infty,j}  \lambda^{\frac{k+j}{2} -2} +  \sum_{j=1}^{2r_\infty-2} \sum_{k=1}^{2r_\infty-2} (-1)^k t_{\infty,j} k I_{\infty,k}\lambda^{\frac{j-k}{2}-2}  + O(\lambda^{-2}) \cr
\end{align}
In order to identify the orders of the determinant, consider first the double time term in the above expansion, one has (omitting the $\frac{1}{4}$ factor)
\begin{align}
\sum_{k=2}^{2r_\infty-3} \sum_{j=1}^{k-1}  (-1)^j t_{\infty,j} t_{\infty,k-j}  \lambda^{\frac{k}{2} -2} + \sum_{k=2r_\infty-2}^{4r_\infty-4} \sum_{j=k - 2r_\infty+2}^{2 r_\infty-2}  (-1)^j t_{\infty,j} t_{\infty,k-j}  \lambda^{\frac{k}{2} -2}  
\end{align}
One realizes that the terms for an odd value of the variable $k$ have no contribution and thus canceling these terms one gets
\begin{align}
   \sum_{k=-1}^{r_\infty-4} \sum_{j=1}^{2 k+3}  (-1)^j t_{\infty,j} t_{\infty,2k+4-j}  \lambda^{k} + \sum_{k= r_\infty-3}^{2 r_\infty-4} \td{P}^{(2)}_{\infty,k}  \lambda^k
\end{align}
The analysis for the second term goes in the same spirit, to recover the result, the reduction from $\mathfrak{gl}_2 (\mathbb{C}) \rightarrow \mathfrak{sl}_2(\mathbb{C})$ must be used. For this, one needs to make sure that the sum of the two eigenvalues $y_1$ and $y_2$ whose expansions are respectively given in \eqref{expansion} is zero (this is because for traceless matrices, the eigenvalues admit opposite expansions). Note that this only possible if the even times (this is also described in the reduction), and the even isospectral Hamiltonians $I_{2 k},$ for $k\in \llbracket 1 , r_\infty-1 \rrbracket$ are all set to zero. Taking this into account, one has the second sum that reduces to 
\begin{align}
    \sum_{i=0}^{2r_\infty-3} \sum_{k=1}^{2 r_\infty -2-i} (-1)^k  k t_{\infty,i+k}  I_{\infty,k} \lambda^{\frac{i}{2}-2}  + \sum_{i=- 2 r_\infty+3}^{-1} \sum_{k=1-i}^{2 r_\infty -2} (-1)^k  k t_{\infty,i+k}  I_{\infty,k} \lambda^{\frac{i}{2}-2} 
\end{align}
Note that one needs the positive integer expansion in the local coordinate in order to match the terms on the isomonodromic side, to this end, the second term contributes with the powers of the local coordinate via
\begin{align}
\sum_{i=0}^{2r_\infty-3} \sum_{k=1}^{2 r_\infty -2-i} (-1)^k  k t_{\infty,i+k}  I_{\infty,k} \lambda^{\frac{i}{2}-2} + O(\lambda^{-1}) =  \sum_{i=2}^{r_\infty-2} \sum_{k=1}^{2 r_\infty -2-2i} (-1)^k  k t_{\infty,2i+k}  I_{\infty,k} \lambda^{i-2}  + O(\lambda^{-1}) \nonumber \\
= - \sum_{i=0}^{r_\infty-4} \sum_{k=1}^{ r_\infty -3-i}  (2k-1) t_{\infty,2i+3+2k}  I_{\infty,2k-1} \lambda^{i}  + O(\lambda^{-1})
\end{align}
where in the last step we have applied the reduction, eliminating the even terms and translated the variable $i \to i+2$. In particular, this allows us to get the following result.
\begin{theorem}[Relation between $\mathbf{I}_{\infty,k}$ and $\mathbf{H}_{\infty,k}$] The identification of the determinant of the connection matrix establishes a direct map between the sets $\mathbf{I}_{\infty,k}$ and $\mathbf{H}_{\infty,k}$ given by 
\begin{align}
    H_{\infty,k} =  \Res_{\lambda \to \infty} \lambda^{-k-1} [\td{L}(\lambda)]_{1,2} \partial_\lambda \frac{[\td{L}(\lambda)]_{1,1}}{[\td{L}(\lambda)]_{1,2}} - \frac{1}{4}\sum_{j=1}^{2 k+3} (-1)^j  t_{\infty,j} t_{\infty,2k+4-j}  + \sum_{j=1}^{ r_\infty -3 -k}  (2j-1) t_{\infty,2k+3 -2j}  I_{\infty,2j-1}
\end{align}
written in a matrix form as 
\small{\beq\label{MatrixIHInfty} \hat{M}_{\infty}(\mathbf{t})\begin{pmatrix}I_{\infty,1}\\3I_{\infty,3}\\\vdots\\ (2r_{\infty}-7)I_{\infty,2r_\infty-7} \end{pmatrix}= \begin{pmatrix}H_{\infty,r_\infty-4}\\ H_{\infty,r_\infty-5}\\\vdots \\ H_{\infty,0}\end{pmatrix}
+N_\infty(\lambda,\mathbf{t},\mathbf{q},\mathbf{p})
\eeq}
with the matrix $\hat{M}_{\infty}(\mathbf{t})$ is given in (\ref{MatrixMInfty}) and reduces under the reduction to
\begin{align}
  \begin{pmatrix}t_{\infty,2r_\infty-3}&0&\dots& &0\\
t_{\infty,2r_\infty-5}& t_{\infty,2r_\infty-3}&0& \ddots& \vdots\\
\vdots& \ddots& \ddots&&\vdots \\
t_{\infty,7}& & \ddots &\ddots&0 \\
t_{\infty,5}& t_{\infty,7} & \dots & &t_{\infty,2r_\infty-3}\end{pmatrix}.
\end{align}
where the last two lines and columns are excluded giving a $(r_\infty - 3) \times (r_\infty - 3) $ system. The matrix $N_\infty(\lambda,\mathbf{t},\mathbf{q},\mathbf{p})$ is given by 
\begin{align}
   N_\infty(\lambda,\mathbf{t},\mathbf{q},\mathbf{p}) := \begin{pmatrix} \frac{1}{4} \underset{j=1}{\overset{2r_\infty-5}{\sum}} (-1)^j t_{\infty,2r_\infty-4-j}t_{\infty,j}-\underset{\lambda\to \infty}{\Res}\lambda^{-r_\infty+3} [\td{L}(\lambda)]_{1,2} \partial_\lambda \frac{[\td{L}(\lambda)]_{1,1}}{[\td{L}(\lambda)]_{1,2}} \\
 \frac{1}{4}\underset{j=1}{\overset{2r_\infty-7}{\sum}} (-1)^j t_{\infty,2r_\infty-6-j}t_{\infty,j}-\underset{\lambda\to \infty}{\Res}\lambda^{-r_\infty+4} [\td{L}(\lambda)]_{1,2} \partial_\lambda \frac{[\td{L}(\lambda)]_{1,1}}{[\td{L}(\lambda)]_{1,2}}  \\
\vdots\\
\frac{1}{4}\underset{j=1}{\overset{3}{\sum}}(-1)^j t_{\infty,4-j}t_{\infty,j}-\underset{\lambda\to \infty}{\Res}\lambda^{-1} [\td{L}(\lambda)]_{1,2} \partial_\lambda \frac{[\td{L}(\lambda)]_{1,1}}{[\td{L}(\lambda)]_{1,2}} 
  \end{pmatrix}
\end{align}
Note in particular that this result is valid only for the canonical choice of trivial times made after the reduction.
\end{theorem}
As a corollary one gets the Hamiltonian system in terms of the isospectral Hamiltonians

\begin{corollary}[Hamiltonian system in $\mathbf{I_{\infty,k}}$]\label{Hamwithiso} The reduction of the symplectic structure and the relation between the sets $\mathbf{I}_{\infty,k}$ and $\mathbf{H}_{\infty,k}$ allows one to express the general Hamiltonian of the Painlev\'{e} I hierarchy using the isospectral Hamiltonians via 
\begin{align}
    \text{Ham}^{(\boldsymbol{\alpha}^\tau)}(\check{\mathbf{q}},\check{\mathbf{p}})=&  \sum_{k=0}^{r_\infty-4} \nu_{\infty,k+1}^{(\boldsymbol{\alpha}^\tau)} \bigg(  \Res_{\lambda \to \infty} \lambda^{-k-1} [\td{L}(\lambda)]_{1,2} \partial_\lambda \frac{[\td{L}(\lambda)]_{1,1}}{[\td{L}(\lambda)]_{1,2}} \nonumber \\
    &- \frac{1}{4} \sum_{j=1}^{2 k+3} (-1)^j  t_{\infty,j} t_{\infty,2k+4-j}  + \sum_{j=1}^{ r_\infty -3 -k}  (2j-1) t_{\infty,2k+3 -2j}  I_{\infty,2j-1} \bigg)
\end{align}
where the coefficients $\nu_{\infty,k+1}^{(\boldsymbol{\alpha}^\tau)}$ are given in \eqref{Matrixentries}.
\end{corollary}

\begin{remark}
     Note that the above expression of the Hamiltonian depends on the choice of coordinates that one uses, in particular, in the next section several sets of coordinates are established in order to define an isospectral set of coordinates and the expression of the Hamiltonian depends on this choice. For instance, if one wishes to express the Hamiltonian in the variables $(\mathbf{Q},\mathbf{R})$ (introduced in the next section), one cannot just perform a simple replacement since the change of coordinates is not symplectic. There are however several ways to reach the Hamiltonian, given a set of coordinates on the connection and auxiliary matrices.    
\end{remark}

\begin{remark}
    In order to be more precise on the reduction of the trace, let us note that the passage $\mathfrak{gl}_2 (\mathbb{C}) \rightarrow \mathfrak{sl}_2(\mathbb{C})$ has direct consequences on the eigenvalues of the matrix $\td{L}$ used in the above identification. This could be realized immediately from the expansion of the eigenvalues since for traceless matrices, one has the characteristic equation that gives $y^2 - \det(\td{L}) = 0$ giving opposite eigenvalues
    \begin{align}
        y = \pm \sqrt{\det (\td{L})}
    \end{align}
    We already know that the determinant is an odd polynomial of the local coordinate, thus the eigenvalues admit in this case an expansion dependent only on half integer powers. This is equivalent to say also that $y_1 + y_2 = 0$. Note also that the connection matrix in different gauges have different eigenvalues due to the additional term $\partial_\lambda G \, G^{-1}$ that is not necessarily traceless. This observation is thus specific to this gauge. 
\end{remark}

\subsection{From isomonodromic to isospectral coordinates} 
In this section, the correspondence between the isomonodromic and the isospectral symplectic structures is established, this is done through the explicit relation between the isomonodromic coordinates and the isospectral ones $(\mathbf{q}, \mathbf{p} ) \to (\mathbf{u}, \mathbf{v})$. The problem of finding an appropriate candidate to trivialize the symplectic structure has always been a problem on the isospectral side, for instance, this issue was left open in \cite{BertolaHarnadHurtubise2022} on the generic side, the twisted side remains less considered. More recently, several techniques were employed to unify both approaches, of particular interest is the one used in \cite{BertolaHarnadHurtubise2022} and applied in \cite{Marchal_2024} to obtain the link between the two setups, this perspective is adopted here. More precisely, the starting point is the algebraic structure of the first line of $\td{L}$, the isospectral coordinates are chosen so that an additional condition is satisfied by the Lax pair $(\td{L},\td{A})$ of the zero curvature equation, the coordinates are then defined to be the set of coordinates for which 
\begin{align}\label{isospecondition}
    \delta^{(\boldsymbol{\alpha})}_{\mathbf{t}} \td{L}(\lambda) = \partial_\lambda \td{A}_{\boldsymbol{\alpha}}(\lambda)
\end{align}
This condition is referred to as the ``\textbf{\textit{isospectral condition}}'' and any set of coordinates satisfying this condition is referred to as a set of ``\textbf{\textit{isospectral coordinates}}''. \\

It is worth mentioning at this stage that the isospectral coordinates are not unique since any symplectic time independent change of coordinates preserves the symplectic structure. In particular, this set of coordinates is not canonical and therefore a time independent but non-symplectic change of coordinates is required to trivialize the symplectic structure. Furthermore, the map linking $(\mathbf{q}, \mathbf{p} ) \to (\mathbf{u}, \mathbf{v})$ cannot be time independent and symplectic since it is required to cancel the additional time dependent terms in \autoref{Hamwithiso}. Note that this is non-trivial since the essence of the problem lies in finding a map that takes a set of coordinates to another ensuring that the additional time dependent terms in the general Hamiltonian cancel out. There is no general method that ensures the existence of such a map and no actual recipe to construct it. This is exactly why instead of tackling the problem directly, one makes use of the explicit geometric construction providing the coordinate dependent formulas for the Lax pair and solves the isospectral condition which provides a map between both sets of coordinates. 

\subsubsection{An intermediate step}

In order to link both sets of coordinates, an intermediate set of coordinates denotes $(Q_j,P_j)_{1 \leq j \leq g}$ is introduced. Motivated by the definition of the \textit{Geometric Coordinates} of \cite{Marchal_2024}, one defines
\begin{definition}[Geometric and Lax coordinates] Define the \textbf{\textit{Geometric set of coordinates}} $(\mathbf{Q}, \mathbf{P})$ by 
\begin{align}
    \prod_{j=1}^g (\lambda-q_j) &:=  \sum_{k=0}^{r_\infty-4 } Q_k \lambda^k + \lambda^{r_\infty-3} \nonumber \\
    p_j &:= \sum_{k=0}^{r_\infty-4 } P_k \frac{\partial Q_k (q_1,\dots,q_g)}{\partial q_j}
\end{align}
These are complemented with an additional set of coordinates called \textit{\textbf{Lax coordinates}} $(\mathbf{Q},\mathbf{R})$ given by 
\begin{align}
    R_{k} :=  P_{r_\infty-4 - k } + \sum_{m=0}^{r_\infty-5-k } P_m Q_{k+1+m} -\delta_{k >0 } \frac{1}{2}t_{\infty,2r_\infty-2}Q_{k-1} - g_0 Q_k
\end{align}  
\end{definition}
The idea behind the first change of coordinates is algebraically natural since it corresponds to writing a function in its factorized or polar forms, depending on the gauge one wishes to use. It is in particular a symplectic change of coordinates holding no time dependence. One direct consequence of the symplectic nature of this change is that one could immediately obtain the Hamiltonians in $(\mathbf{Q}, \mathbf{P})$ just by a replacement (this is not the case for the second set of coordinates as the change is not symplectic), however, this strategy is not practical in general. The alternative strategy is to make use of the connection matrix and its auxiliary counter part. The first change od coordinates is constructed to express specifically the entry $[\td{L}]_{1,2}$ with a symplectic change. The second change of coordinates is chosen specifically to express the entry $[\td{L}]_{1,1}$ in an algebraic form. This change of coordinates is generally not symplectic and is time dependent. 

\medskip

Our goal is to evaluate the isospectral condition and for this, we provide the necessary entries of the connection and auxiliary matrices.

\begin{proposition}[Matrices using $(\mathbf{Q}, \mathbf{P})$]\label{geometricmatrices} The connection matrix is given, in terms of the coordinates $(\mathbf{Q}, \mathbf{P})$, by
\begin{align}
    [\td{L} (\lambda)]_{1,2} = & \sum_{k=0}^{r_\infty-4 } Q_k \lambda^k + \lambda^{r_\infty-3} \nonumber \\
     [\td{L} (\lambda)]_{1,1} = & \sum_{k=0}^{r_\infty-4 } P_{r_\infty-4-k}\lambda^k + \sum_{k=0}^{r_\infty-5 } \sum_{m=0}^{r_\infty-5-k } P_m Q_{k+1+m} \lambda^k  -\left(\frac{1}{2}t_{\infty,2r_\infty-2}\lambda+g_0 \right)  \left(\sum_{k=0}^{r_\infty-4 } Q_k \lambda^k + \lambda^{r_\infty-3} \right)   \nonumber  
\end{align}
While, its auxiliary counter part has the following expression 
\begin{align}
    [\td{A}_{\boldsymbol{\alpha}} (\lambda)]_{1,2} = &\nu_{\infty,-1}^{(\boldsymbol{\alpha})}  \lambda^{r_\infty-2} + \sum_{j=0}^{r_\infty-3} \left( \nu_{\infty,r_\infty-3-j}^{(\boldsymbol{\alpha})} + \sum_{k=j-1}^{r_\infty-4}Q_{\infty,k}\nu_{\infty,k-j}^{(\boldsymbol{\alpha})} \right) \lambda^j \nonumber \\
     [\td{A}_{\boldsymbol{\alpha}} (\lambda)]_{1,1} =&  -\frac{1}{2} t_{\infty,2r_\infty-2} \sum_{j=0}^{r_\infty-3}  \nu^{(\boldsymbol{\alpha})}_{\infty,r_\infty-2-j} \lambda^{j} +\sum_{j=0}^{r_\infty-5}  \nu^{(\boldsymbol{\alpha})}_{\infty,r_\infty-4-j} \left( P_{0} - \frac{1}{2} t_{\infty,2r_\infty-2}Q_{r_\infty-5} + g_0 Q_{r_\infty-4}   \right) \lambda^{j }+ \nonumber \\
     &\sum_{j=0}^{r_\infty-6} \sum_{i=1}^{r_\infty-5-j} \nu^{(\boldsymbol{\alpha})}_{\infty,i} \left(  P_{r_\infty-4-i-j}\lambda^k + \sum_{m=0}^{r_\infty-5-k } P_m Q_{i+j+1+m} - Q_{i+j} g_0 -\frac{1}{2} t_{\infty,2r_\infty-2} Q_{i+j-1} \right) \lambda^j \nonumber 
\end{align}
\end{proposition}
\begin{proof}
    The proof is done in \autoref{Proofproposition}. The entry $(2,1)$ is omitted to keep the paper as concise as possible, as it does not contribute to the analysis leading to the isospectral coordinates as shown in \cite{Marchal_2024} (the strategy to get this entry follows in the same way as in op. cit. relying on the gauge transformation so one encounters no particular problems for the second line of the matrices). 
\end{proof}

These expressions provide half the isospectral coordinates, their associated duals are established using the Lax coordinates expressions.  

\begin{proposition}\label{Laxcoordinates}[Matrices using Lax coordinates $(\mathbf{Q}, \mathbf{R})$] Using the set of Lax coordinates $(\mathbf{Q}, \mathbf{R})$, the connection matrix admits the following expansion 
\begin{align}
    [\td{L} (\lambda)]_{1,2} = &  \sum_{k=0}^{r_\infty-4 } Q_k \lambda^k + \lambda^{r_\infty-3}  \nonumber \\
    [\td{L} (\lambda)]_{1,1} = & \sum_{k=0}^{r_\infty-4} R_k \lambda^k - \frac{1}{2} \lambda^{r_\infty-3} - \frac{1}{2}t_{\infty,2r_\infty-2}\lambda^{r_\infty-2}  \nonumber 
\end{align}
While its auxiliary counter part is given by 
\begin{align}
    [\td{A}_{\boldsymbol{\alpha}} (\lambda)]_{1,2} = & \nu_{\infty,1}^{(\boldsymbol{\alpha})} \lambda^{r_\infty-4}+\sum_{j=0}^{r_\infty-5}\left( \nu_{\infty,r_\infty-3-j}^{(\boldsymbol{\alpha})}+\sum_{k=j+1}^{r_\infty-4}\nu_{\infty,k-j}^{(\boldsymbol{\alpha})}Q_{\infty,k}\right)\lambda^j +O(\lambda^{-1}) \nonumber \\
    [\td{A}_{\boldsymbol{\alpha}} (\lambda)]_{1,1} = & \sum_{j=0}^{r_\infty-5} \left( \sum_{i=1}^{r_\infty-4 - j }  R_{j+i} \nu^{(\boldsymbol{\alpha})}_{\infty,i} \right) \lambda^j  - \frac{1}{2} \sum_{j=0}^{r_\infty-4}  \nu^{(\boldsymbol{\alpha})}_{\infty,r_\infty-3-j}\lambda^{j} \nonumber \\
   &  - \frac{1}{2}t_{\infty,2r_\infty-2}  \sum_{j=0}^{r_\infty-3}  \nu^{(\boldsymbol{\alpha})}_{\infty,r_\infty-2-j}\lambda^{j}+ O(\lambda^{-1}) \nonumber
\end{align}    
\end{proposition}
\begin{proof}
   The computation for the entries of both matrices is rather straightforward, this is why in \cite{Marchal_2024} the authors named these coordinates the Lax coordinates. The entry $(1,1)$ is the only non-trivial entry, it is detailed in \autoref{proofLaxcoordinates}.  
\end{proof}

\subsubsection{The isospectral coordinates}

Now that one has the explicit expressions, the strategy consists in using this new set to try and solve the isospectral condition (\ref{isospecondition}) for the associated matrices and obtain a sufficient condition that coordinates in the isospectral approach must satisfy. Starting from entry $(1,2)$ one has
\begin{align}
    \sum_{k=0}^{r_\infty-4} \delta_{\mathbf{t}}^{(\boldsymbol{\alpha})} [Q_k] \lambda^k =(r_\infty-4) \nu_{\infty,1}^{(\boldsymbol{\alpha})}  \lambda^{r_\infty-5} \delta_{r_\infty \geq 5} + \sum_{j=0}^{r_\infty-5}j \left( \nu_{\infty,r_\infty-3-j}^{(\boldsymbol{\alpha})} + \sum_{k=j+1}^{r_\infty-4}Q_{\infty,k}\nu_{\infty,k-j}^{(\boldsymbol{\alpha})} \right) \lambda^{j-1} \nonumber
\end{align}
where we have omitted the additional $O(\lambda^{-1})$ terms since they play no role. This expression provides the following result 
\begin{theorem}[Isospectral coordinates $\mathbf{u}$]\label{Theo4.2} Half of the isospectral coordinates are defined via the isospectral condition $\delta_{\mathbf{t}}^{(\boldsymbol{\alpha})} [\td{L}]_{1,2} = \partial_\lambda [\td{A}_{\boldsymbol{\alpha}} ]_{1,2}$ as the coordinates satisfying the differential system 
\footnotesize{\beq \mathbf{T} (\mathbf{t}) \begin{pmatrix} \delta^{(\boldsymbol{\alpha})}_{\mathbf{t}}\left[\frac{Q_{\infty,r_\infty-5}}{r_\infty-4}\right]\\
\vdots\\
\delta^{(\boldsymbol{\alpha})}_{\mathbf{t}}\left[\frac{Q_{\infty,0}}{1}\right]\end{pmatrix}=
\begin{pmatrix}1&0&\dots&0\\
Q_{\infty,r_\infty-4}&1& &0\\
\vdots& \ddots &\ddots &\vdots \\
Q_{\infty,2}&\dots&Q_{\infty,r_\infty-4}&1
\end{pmatrix}
\begin{pmatrix} \frac{2\alpha_{\infty,2r_\infty-7}}{(2r_\infty-7)}\\ \frac{2\alpha_{\infty,2r_\infty-9}}{(2r_\infty-9)}  \\ \vdots \\ \frac{2\alpha_{\infty,3}}{3} \end{pmatrix} 
\eeq}
where the time dependent matrix is given by 
\begin{align}
  \mathbf{T}(\mathbf{t} ) :=  \begin{pmatrix}2&0&\dots&& &0\\
0& 2&0& \ddots& &\vdots\\
t_{\infty,2r_\infty-7}&\ddots&\ddots& \ddots&&\vdots\\
\vdots& \ddots& \ddots&&&\vdots \\
t_{\infty,9}& & \ddots &\ddots&&0 \\
t_{\infty,7}& t_{\infty,7} & \dots & t_{\infty,2r_\infty-7}& 0&2\end{pmatrix}
\end{align}
The first coordinate is immediately given by $Q_{\infty,r_\infty-4} = u_{r_\infty-4}$. In particular, their existence is manifested by the existence of a solution of the above system realized due to the recursive Toeplitz structure and the compatibility of the system. 
    
\end{theorem}
\begin{proof}
    One realizes immediately that the identification of the orders of the entry $(2,1)$ of the isospectral condition provides 
    \begin{align} \label{Firstconditions}
        \delta_{\mathbf{t}}^{(\boldsymbol{\alpha})} [Q_{r_\infty-4}] = &\,0 \cr
        \delta_{\mathbf{t}}^{(\boldsymbol{\alpha})} [Q_{r_\infty - 5}] = & (r_\infty-4) \nu_{\infty,1}^{(\boldsymbol{\alpha})}  \cr
        \delta_{\mathbf{t}}^{(\boldsymbol{\alpha})} [Q_{m}] = & (m+1) \left( \nu_{\infty,r_\infty-4-m}^{(\boldsymbol{\alpha})}+\sum_{k=m+2}^{r_\infty-4}\nu_{\infty,k-m-1}^{(\boldsymbol{\alpha})}Q_{\infty,k}\right), \,\,\,\,\, \forall \, m \, \in \llbracket 0, r_\infty - 6 \rrbracket
    \end{align}
    Suppose that $r_\infty \geq 5$ which is the value for which the analysis becomes non-trivial, let us note that the situation at the infinite pole is dependent on the chosen normalization of the connection matrix. The first line of the above system trivially implies $Q_{r_\infty-4}= u_{r_\infty-4}  $ giving the first isospectral coordinate. The remaining two lines could be written as follows
    \beq \begin{pmatrix} \delta^{(\boldsymbol{\alpha})}_{\mathbf{t}}\left[\frac{Q_{r_\infty-5}}{r_\infty-4}\right]\\
    \vdots\\
     \delta^{(\boldsymbol{\alpha})}_{\mathbf{t}}\left[\frac{Q_{0}}{1}\right]\end{pmatrix}=\begin{pmatrix}1&0&\dots&0\\
 Q_{r_\infty-4}&1& &0\\
\vdots& \ddots &\ddots &\vdots \\
Q_{2}&\dots&Q_{r_\infty-4}&1
\end{pmatrix}\begin{pmatrix}\nu_{\infty,1}^{(\boldsymbol{\alpha})}\\ \vdots\\ \nu_{\infty,r_\infty-4}^{(\boldsymbol{\alpha})}\end{pmatrix}
\eeq
The next step is to use the explicit expression of the $\nu_{\infty,i}^{(\boldsymbol{\alpha})}$, which gives 
\footnotesize{\beq \label{system} \begin{pmatrix} \delta^{(\boldsymbol{\alpha})}_{\mathbf{t}}\left[\frac{Q_{\infty,r_\infty-5}}{r_\infty-4}\right]\\
\vdots\\
\delta^{(\boldsymbol{\alpha})}_{\mathbf{t}}\left[\frac{Q_{\infty,0}}{1}\right]\end{pmatrix}=
\begin{pmatrix}1&0&\dots&0\\
Q_{\infty,r_\infty-4}&1& &0\\
\vdots& \ddots &\ddots &\vdots \\
Q_{\infty,2}&\dots&Q_{\infty,r_\infty-4}&1
\end{pmatrix}\mathbf{T}(\mathbf{t})^{-1}
\begin{pmatrix}\frac{2\alpha_{\infty,2r_\infty-7}}{(2r_\infty-7)}\\ \frac{2\alpha_{\infty,2r_\infty-9}}{(2r_\infty-9)}\\ \vdots \\ \frac{2\alpha_{\infty,3}}{3} \end{pmatrix}
\eeq}
At this stage, one uses the commutation of lower-triangular Toeplitz matrices and eliminates the factors $\alpha_{\infty,i}$ by restricting the derivative operator to $\delta^{(\boldsymbol{\alpha})}_{\mathbf{t}} = \delta_t$. The first line of the above system is equivalent to 
\begin{align}
    \frac{2}{r_\infty-4} \delta_{t_{\infty,k}} [Q_{\infty,r_\infty-5}] = \delta_{k,2 r_\infty-7} \frac{2}{2 r_\infty-7}
\end{align} 
giving 
\begin{align}
    Q_{\infty,r_\infty-5} = \frac{ r_\infty-4}{2 r_\infty-7} t_{\infty,2 r_\infty-7} + u_{r_\infty-5}
\end{align}
with $u_{r_\infty-5}$ independent of all the irregular times. It is clear that the further one goes with the lines the more involved the expressions become. The proof of the theorem is complete once the existence of a solution is established, note that this is not a trivial result despite the possibility of writing the system simply as  $\delta^{(\boldsymbol{\alpha})}_{\mathbf{t}} \mathcal{Q} = \mathcal{F }(\mathcal{Q},\mathbf{t})$ . This is due to the multi-component differential $\delta^{(\boldsymbol{\alpha})}_{\mathbf{t}}$. This system is a system of first order PDEs, in particular, for a given solution to exist, one must ensure the commutation of the operators $ [ \delta_{t_i}, \delta_{t_j} ]Q_k = 0$. This commutation relation is sufficient for the purpose present here (since the isospectral coordinates are assumed independent of any irregular time). This is going to be the main purpose of \autoref{AppD} and this ends the proof.
\end{proof}
Let us note that one could write a matrix solution of the above system for the first few cases, yet the expressions become quickly involved (yet an appropriate software could handle the amount of computation for a fixed value of the irregular pole), this is what the authors of \cite{Marchal_2024} realized on the generic side of the problem. Note that this is unsurprising since getting the explicit time dependent factors, the Hamiltonians of the underlying integrable systems identify with the set $\mathbf{I}_{\infty,k}$, a problem that is known to be substantially difficult.   
 \begin{remark}
     In \cite{Marchal_2024}, the authors considered an additional parameter $\omega$ that appeared in the correspondence on the generic side, note that this was essential since the case $r_\infty=1$ of the generic side was covered and the factor $\omega$ plays a crucial role in that case. However, in the setup of this paper, a suitable normalization sets this factor immediately to $\omega=1$. 
 \end{remark}
This result defines only half of the isospectral coordinates, the other half is given by the following result

\begin{theorem}[Isospectral coordinates $\mathbf{v}$]\label{Theo4.3} One obtains the second half of the isospectral coordinates by evaluating the isospectral condition $\delta_{\mathbf{t}}^{(\boldsymbol{\alpha})} [\td{L}]_{1,1} = \partial_\lambda [\td{A}_{\boldsymbol{\alpha}} ]_{1,1}$ as the coordinates satisfying the differential system 

\footnotesize{\beq \mathbf{T(t)} \begin{pmatrix} \delta^{(\boldsymbol{\alpha})}_{\mathbf{t}}\left[\frac{R_{\infty,r_\infty-5}}{r_\infty-4}\right]\\
\vdots\\
\delta^{(\boldsymbol{\alpha})}_{\mathbf{t}}\left[\frac{R_{\infty,0}}{1}\right]\end{pmatrix}=
\begin{pmatrix}-1&0&\dots&0\\
R_{\infty,r_\infty-4}&-1& &0\\
\vdots& \ddots &\ddots &\vdots \\
R_{\infty,2}&\dots&R_{\infty,r_\infty-4}&-1
\end{pmatrix} 
\begin{pmatrix}\frac{\alpha_{\infty,2r_\infty-7}}{(2r_\infty-7)}\\ \frac{\alpha_{\infty,2r_\infty-9}}{(2r_\infty-9)}\\ \vdots \\ \frac{\alpha_{\infty,3}}{3} \end{pmatrix}
\eeq}
with the first coordinate given by $R_{r_\infty-4} = v_{r_\infty-4}$. In particular, the existence of these coordinates is manifested by the compatibility of the system and the recursive solution due to the lower triangular Toeplitz structure of the system, completing the isospectral coordinates. 
\end{theorem}
\begin{proof}
The strategy consists in evaluating the derivative operators immediately on the entries of the connection and auxiliary matrices, this provides following an identification of the orders of
\begin{align}
    \sum_{k=0}^{r_\infty-4} \delta^{(\boldsymbol{\alpha})}_{\mathbf{t}} [R_k] \lambda^k = \sum_{j=0}^{r_\infty-5} (j+1) \left( \sum_{i=1}^{r_\infty-5 - j }  R_{j+i+1} \nu^{(\boldsymbol{\alpha})}_{\infty,i} \right) \lambda^{j}  - \frac{1}{2}  \sum_{j=0}^{r_\infty-5} (j+1) \nu^{(\boldsymbol{\alpha})}_{\infty,r_\infty-4-j}\lambda^{j}+ O(\lambda^{-1})
\end{align}
the first coordinate given by $R_{r_\infty-4} = v_{r_\infty-4}$, the other powers give 
\begin{align}
    \frac{1}{k+1} \delta^{(\boldsymbol{\alpha})}_{\mathbf{t}} [R_k] =  \sum_{i=1}^{r_\infty-5 - k }  R_{k+i+1} \nu^{(\boldsymbol{\alpha})}_{\infty,i} - \frac{1}{2} \nu^{(\boldsymbol{\alpha})}_{\infty,r_\infty-4-k}
\end{align}
Restricting to the reduced case (even times are set to 0 to eliminate the trace), and restricting $\delta^{(\boldsymbol{\alpha})}_{\mathbf{t}}$ to $\delta_t$ eliminating the $\alpha$ terms. One transforms the system
\beq \begin{pmatrix} \delta^{(\boldsymbol{\alpha})}_{\mathbf{t}}\left[\frac{R_{r_\infty-5}}{r_\infty-4}\right]\\
    \vdots\\
     \delta^{(\boldsymbol{\alpha})}_{\mathbf{t}}\left[\frac{R_{0}}{1}\right]\end{pmatrix}=\begin{pmatrix}-\frac{1}{2}&0&\dots&0\\
 R_{r_\infty-4}&-\frac{1}{2}& &0\\
\vdots& \ddots &\ddots &\vdots \\
R_{2}&\dots&R_{r_\infty-4}&- \frac{1}{2}
\end{pmatrix}\begin{pmatrix}\nu_{\infty,1}^{(\boldsymbol{\alpha})}\\ \vdots\\ \nu_{\infty,r_\infty-4}^{(\boldsymbol{\alpha})}\end{pmatrix}
\eeq
to the system 
\footnotesize{\beq \begin{pmatrix} \delta^{(\boldsymbol{\alpha})}_{\mathbf{t}}\left[\frac{R_{\infty,r_\infty-5}}{r_\infty-4}\right]\\
\vdots\\
\delta^{(\boldsymbol{\alpha})}_{\mathbf{t}}\left[\frac{R_{\infty,0}}{1}\right]\end{pmatrix}=
\begin{pmatrix}-1&0&\dots&0\\
R_{\infty,r_\infty-4}&-1& &0\\
\vdots& \ddots &\ddots &\vdots \\
R_{\infty,2}&\dots&R_{\infty,r_\infty-4}&-1
\end{pmatrix}\mathbf{T(t)}^{-1}
\begin{pmatrix}\frac{\alpha_{\infty,2r_\infty-7}}{(2r_\infty-7)}\\ \frac{\alpha_{\infty,2r_\infty-9}}{(2r_\infty-9)}\\ \vdots \\ \frac{\alpha_{\infty,3}}{3} \end{pmatrix}
\eeq}
Using the commutation of the lower triangular Toeplitz matrices, one obtains the result of the theorem. At the end, one could obtain the explicit recursive solution of this system in a similar way to what is shown in \autoref{Theo4.2} despite the fact that the system is relatively more complicated. On the other side, as in the previous theorem, one needs to ensure the compatibility of the system by ensuring the commutation of the differential operator on the coordinates $ [ \delta_{t_i}, \delta_{t_j} ]R_k = 0$. For this, a similar analysis to the one done in \autoref{AppD} could be carried out for this set of coordinates.  
\end{proof}

The above construction leads to the proof of \autoref{Principalresult}
\begin{proof}
    The construction of the present paper along with the generic case treated in \cite{Marchal_2024} offer a map from the isomonodromic coordinates, taken as the apparent singularities, to the isospectral coordinates satisfying the isospectral condition (\ref{isospecondition}) given explicitly by \autoref{Theo4.2} and \autoref{Theo4.3}. In particular, this map is time dependent and non-symplectic and manifested by two changes of coordinates established for both twisted and generic rank $2$ setups. Note that the entry $(2,1)$ plays no role since this entry is chosen so that the zero curvature equation is satisfied, solving it would imply a consistency result rather than a new set of coordinates. 
\end{proof}

\appendix
\section{Proof of \autoref{Sym2form}} \label{proof}
This appendix is devoted to the proof of the symplectic reduction property denoted in \autoref{Sym2form}. The strategy of the proof consists on using differential geometry calculus in order to compute and obtain the final result. Let us first note that following \autoref{ShiftDarbouxCoordinates}, one has
\begin{align}
    \sum_{j=1}^g d\check{q}_j \wedge d \check{p}_j =   \sum_{j=1}^g d q_j \wedge d p_j 
\end{align}
in other words, the change of coordinates is a symplectic change. On the other hand one may write from the evolution of the coordinates the expressions 
\begin{align}
    \text{Ham}^{(\mathbf{w}_k)} (\mathbf{q,p})= - \sum_{j=1}^g  \frac{q_j^k}{2 k}, \qquad
    \text{Ham}^{(\mathbf{u}_{-1})}(\mathbf{q,p}) =  - \sum_{j=1}^g q_j p_j, \qquad
    \text{Ham}^{(\mathbf{u}_0)}(\mathbf{q,p}) =  \sum_{j=1}^g q_j-p_j
\end{align}
Note in particular that one has $\mathcal{L}_{\mathbf{w}_k}[T_1]=\mathcal{L}_{\mathbf{w}_k}[T_2]=\mathcal{L}_{\mathbf{w}_k}[q_j]=0$, and thus immediately $\mathcal{L}_{\mathbf{w}_k}[\check{q}_j]=0$. Furthermore, 
\beq \mathcal{L}_{\mathbf{w}_k}[\check{p}_j]=T_2^{-1}\left(\mathcal{L}_{\mathbf{w}_k}-\frac{1}{2}\mathcal{L}_{\mathbf{w}_k}[\td{P}_1](q_j) \right)
=T_2^{-1}\left(-\frac{\hbar}{2} q_j^{k-1}+\frac{1}{2}\hbar q_j^{k-1}\right)=0
\eeq
This immediately implies that the change of coordinates implies $   \text{Ham}^{(\mathbf{w}_k)} (\mathbf{\check{q},\check{p}})=0$. This first step eliminates the $\mathbf{e}_{2k}$ from the second term of the extended fundamental symplectic $2-$form. Using the same argument, we also eliminate the Hamiltonian in the $\mathbf{u}_{-1}$ and $\mathbf{u}_0$ directions, for that one writes 
\small{\bea \mathcal{L}_{\mathbf{u}_{-1}}[\check{q}_j]&=&\mathcal{L}_{\mathbf{u}_{-1}}[T_2]q_j+ T_2\mathcal{L}_{\mathbf{u}_{-1}}[q_j]+\mathcal{L}_{\mathbf{u}_{-1}}[T_1]=\hbar T_2q_j+T_2(-\hbar q_j)=0\cr
\mathcal{L}_{\mathbf{u}_{-1}}[\check{p}_j]&=&-\frac{\mathcal{L}_{\mathbf{u}_{-1}}[T_2]}{T_2^2}\left(p_j-\frac{1}{2}\td{P}_1(q_j)\right)+T_2^{-1}\left(\mathcal{L}_{\mathbf{u}_{-1}}[p_j]-\frac{1}{2}\mathcal{L}_{\mathbf{u}_{-1}}[\td{P}_1](q_j) -\frac{1}{2}\mathcal{L}_{\mathbf{u}_{-1}}[q_j]\td{P}_1'(q_j)\right)\cr
&=&-\hbar T_2^{-1}\left(p_j-\frac{1}{2}\td{P}_1(q_j)\right)\cr
&&+\hbar T_2^{-1}\left(p_j+\frac{1}{4}\sum_{i=0}^{r_\infty-2}(2i+2)t_{\infty,2i+2}q_j^{i} -\frac{1}{2}q_j \sum_{i=1}^{r_\infty-2}it_{\infty,2i+2}q_j^{i-1}\right)\cr
&=&\hbar T_2^{-1}\left( -\frac{1}{2}\sum_{i=0}^{r_\infty-2}t_{\infty,2i+2}q_j^{i}+\frac{1}{4}\sum_{i=0}^{r_\infty-2}(2i+2)t_{\infty,2i+2}q_j^{i}-\frac{1}{2} \sum_{i=1}^{r_\infty-2}it_{\infty,2i+2}q_j^{i}\right)\cr
&=&0
\eea}
On the other hand, one has 
\bea \mathcal{L}_{\mathbf{u}_{0}}[\check{q}_j]&=&\mathcal{L}_{\mathbf{u}_{0}}[T_2]q_j+ T_2\mathcal{L}_{\mathbf{u}_{0}}[q_j]+\mathcal{L}_{\mathbf{u}_{0}}[T_1]=- \hbar T_2+\hbar T_2=0\cr
\mathcal{L}_{\mathbf{u}_{0}}[\check{p}_j]&=&-\frac{\mathcal{L}_{\mathbf{u}_{0}}[T_2]}{T_2^2}\left(p_j-\frac{1}{2}\td{P}_1(q_j)\right)+T_2^{-1}\left(\mathcal{L}_{\mathbf{u}_{0}}[p_j]-\frac{1}{2}\mathcal{L}_{\mathbf{u}_{0}}[\td{P}_1](q_j) -\frac{1}{2}\mathcal{L}_{\mathbf{u}_{0}}[q_j]\td{P}_1'(q_j)\right)\cr
&=&0+T_2^{-1}\left(0+\frac{\hbar}{2} \sum_{s=1}^{r_\infty-2}st_{\infty,2s+2}q_j^{s-1} -\frac{\hbar}{2} \sum_{i=1}^{r_\infty-2}it_{\infty,2i+2}q_j^{i-1}\right)\cr
&=&0
\eea
Let us also denote the time differentials needed 
\begin{align}
    d t_{\infty,2k} = & dT_{\infty,k} \cr
     d t_{\infty, 2r_\infty -3} = & \left( 2 r_\infty -3 \right)T_2^{\frac{2r_\infty -5}{2}}    d T_2 \cr
      dt_{\infty,2r_\infty-5}=&(2r_\infty-5)\,T_2^{\frac{2r_\infty-5}{2}} dT_1 + \frac{(2r_\infty-5)^2}{2} T_1 T_2^{\frac{2 r_\infty-7}{2}} dT_2 
\end{align}
Those are easily obtained from the inverse relations of the times, the challenging part is perhaps the relation for all $k\in \llbracket 1,r_\infty-3\rrbracket$:
\beq t_{\infty,2k-1}=2T_2^{\frac{2k-1}{2}}\left(\sum_{p=1}^{r_\infty-k-2}  \frac{\underset{m=p+1}{\overset{r_\infty-k-2}{\prod}}(2r_\infty-2m-5)}{2^{r_\infty-k-p-2}(r_\infty-k-p-2)!}T_1^{r_\infty-k-p-2}\tau_p + T_1^{r_\infty-1-k}\frac{\underset{m=0}{\overset{r_\infty-k-2}{\prod}}(2r_\infty-2m-5)}{2^{r_\infty-1-k}(r_\infty-1-k)!}
\right)\eeq
which allows one to complete the set of time differentials. indeed for all $k\in \llbracket 1,r_\infty-3\rrbracket$, one has
\begin{align}
     d t_{\infty,2k-1} =& (2k-1) T_2^{\frac{2k-3}{2}}d T_2 \left( \sum_{p=1}^{r_\infty-k-2} \delta(k,p,m)\,\, T_1^{r_\infty-k-p-2}\tau_p + T_1^{r_\infty-1-k}\gamma(k,m) \right) \cr
     +&  2T_2^{\frac{2k-1}{2}} \bigg(  \sum_{p=1}^{r_\infty-k-2} \delta(k,p,m)\left( T_1^{r_\infty-k-p-2} d\tau_p + (r_\infty-k-p-2)T_1^{r_\infty-k-p-3} \tau_p dT_1 \right)   \cr
     +&  (r_\infty-1-k) T_1^{r_\infty-2-k}\gamma(k,m) dT_1  \bigg)   
\end{align}
where we have denoted 
\begin{align}
   \delta (k,p,m)=  \frac{\underset{m=p+1}{\overset{r_\infty-k-2}{\prod}}(2r_\infty-2m-5)}{2^{r_\infty-k-p-2}(r_\infty-k-p-2)!}, \qquad \text{and} \qquad \gamma(k,m)= \frac{\underset{m=0}{\overset{r_\infty-k-2}{\prod}}(2r_\infty-2m-5)}{2^{r_\infty-1-k}(r_\infty-1-k)!}.
\end{align}
Now that one has the full set of differentials, one needs to relate the Hamiltonians following different directions in the basis $\mathbf{e}_k$ and $\mathcal{U}_{trivial},\mathcal{U}_{iso}$. For this, one immediately has $\text{Ham}^{(\mathbf{e}_{2k})}= \text{Ham}^{(\mathbf{w}_k)}$ since $\partial_{\mathbf{e}_{2k}}= \partial_{\mathbf{w}_k}$. A similar, yet more computational argument follows for the other directions in which the Hamiltonians $\text{Ham}^{(\mathbf{\alpha})}(\mathbf{\check{q},\check{p}})=0$, however, the remaining part will be given in the direction $\partial_{\tau_k}$. Note in particular that one has for all $k\in \llbracket 1,r_\infty-3 \rrbracket$ 
\begin{align}
  \partial_{\tau_{k}}=2\sum_{i=1}^{r_\infty-2-k} \delta(i,p,m) T_1^{r_\infty-i-k-2}  T_2^{\frac{2i-1}{2}}\partial_{t_{\infty,2i-1}}
\end{align}
Translating this into the Hamiltonian provides $ \boldsymbol{\alpha}^{\tau_k}_{\infty,2i-1}  \text{Ham}^{(\mathbf{e}_{2i-1})}= \text{Ham}^{(\mathbf{\boldsymbol{\alpha}}_{k})}$ with 
\begin{align}
     \boldsymbol{\alpha}^{\tau_k}_{\infty,2i-1}  = 2\sum_{i=1}^{r_\infty-2-k} \delta(i,p,m) T_1^{r_\infty-i-k-2}  T_2^{\frac{2i-1}{2}}
\end{align}
Let us before addressing the explicit expression denote the following lemma: 
\begin{lemma}\label{DiffCalculus} Denote the basis vectors of the tangent space by the vector fields $\partial_{\mathbf{e}}$ and their dual basis of $1-$forms by $d \mathbf{e}$, let $\tau$ be a function of $\mathbf{e}$, $\alpha$ a function of part of the basis $\mathbf{e}$ denoted $\mathbf{U}$ and $H$ a function expressed by a the vector fields generated by the dual basis to the complement of $\mathbf{U}$. Then one has 
\begin{align}
   d \tau \wedge d (\alpha H) = & d \tau \wedge (d \alpha . H + \alpha . d H ) \cr
   = & d \tau \wedge \alpha . d H. 
\end{align}
This is due to the fact that $ d e_i  \partial_{e_j} = \delta_{ij}$. 
\end{lemma}

Plugging all the ingredients in the second part of the fundamental $2-$form, while eliminating the zero directions, provides 
\begin{align}
    \sum_{k=1}^{2 r_\infty-2} dt_{\infty,k} \wedge d\text{Ham}^{(\mathbf{e}_k)} = &\sum_{k=1}^{ r_\infty-3} dt_{\infty,2k-1} \wedge d\text{Ham}^{(\mathbf{e}_{2k-1})} \cr
   = & \sum_{k=1}^{ r_\infty-3} \bigg( (2k-1) T_2^{\frac{2k-3}{2}}d T_2 \left( \sum_{p=1}^{r_\infty-k-2} \delta(k,p,m)\,\, T_1^{r_\infty-k-p-2}\tau_p + T_1^{r_\infty-1-k}\gamma(k,m) \right) \cr
     +&  2T_2^{\frac{2k-1}{2}} \bigg(  \sum_{p=1}^{r_\infty-k-2} \delta(k,p,m)\left( T_1^{r_\infty-k-p-2} d\tau_p + (r_\infty-k-p-2)T_1^{r_\infty-k-p-3} \tau_p dT_1 \right)   \cr
     +&  (r_\infty-1-k) T_1^{r_\infty-2-k}\gamma(k,m) dT_1  \bigg)     \bigg) \wedge d\text{Ham}^{(\mathbf{e}_{2k-1})}
\end{align}
Note that the only contribution to the wedge product comes from the differential $1-$form associated to the isomonodromic time $\tau_k$, this is due to the fact that the Hamiltonians in the direction $\mathbf{e}_{2k+1}$ hold no dependence on these trivial times. The factor multiplying $d \tau_k$ in the expansion is exactly the term $\boldsymbol{\alpha}^{\tau_k}_{\infty,2i-1}$ in the expansion of the Hamiltonian, this means that we have the part 
\begin{align}
 \sum_{k=1}^{ r_\infty-3}   \boldsymbol{\alpha}^{\tau_k}_{\infty,2i-1}   d \tau_k \wedge d\text{Ham}^{(\mathbf{e}_{2k-1})}
\end{align}
and by \autoref{DiffCalculus}, this is exactly the result needed.

\section{Proof of \autoref{geometricmatrices}}\label{Proofproposition}
\subsection{Entries $(1,2)$}
 First, let us note that the entry $[1,2]$ follows immediately from the change of coordinates, furthermore, since one has from the gauge transformation
\begin{align}
    G(\lambda) \td{A}_{\boldsymbol{\alpha}}(\lambda ) = A_{\boldsymbol{\alpha}} (\lambda) G(\lambda) - \mathcal{L}_{\boldsymbol{\alpha}} [G (\lambda)]
\end{align}
one immediately gets the entry of the auxiliary matrix to be 
\begin{align}
    [ \td{A}_{\boldsymbol{\alpha}}]_{1,2} = & [\td{L}]_{1,2} [ A_{\boldsymbol{\alpha}}]_{1,2} \nonumber \\
    =& \left( \sum_{k=0}^{r_\infty-4} Q_k \lambda^k + \lambda^{r_\infty-3} \right) \left( \sum_{i=-1}^{r_\infty-3} \frac{\nu^{(\boldsymbol{\alpha})}_{\infty,i}}{\lambda^i} +O\left(\lambda^{-(r_\infty-2)}\right)\right) \nonumber \\
    =& \sum_{j=0}^{r_\infty-4}\nu_{\infty,r_\infty-3-j}^{(\boldsymbol{\alpha})} \lambda^{j}+ \sum_{j=0}^{r_\infty-5}\sum_{k=j+1}^{r_\infty-4}Q_{\infty,k}\nu_{\infty,k-j}^{(\boldsymbol{\alpha})} \lambda^j + O\left(\lambda^{-1}\right) \cr
  = &  \nu_{\infty,1}^{(\boldsymbol{\alpha})}  \lambda^{r_\infty-4} + \sum_{j=0}^{r_\infty-5} \left( \nu_{\infty,r_\infty-3-j}^{(\boldsymbol{\alpha})} + \sum_{k=j+1}^{r_\infty-4}Q_{\infty,k}\nu_{\infty,k-j}^{(\boldsymbol{\alpha})} \right) \lambda^j
\end{align}
Where, we have used the fact that 
\begin{align}
    \left[A_{\boldsymbol{\alpha}}(\lambda)\right]_{1,2}\overset{\lambda\to \infty}{=}\sum_{i=1}^{r_\infty-3} \frac{\nu^{(\boldsymbol{\alpha})}_{\infty,i}}{\lambda^i} +O\left(\lambda^{-(r_\infty-2)}\right)
\end{align}
since the terms $\nu^{(\boldsymbol{\alpha})}_{\infty,-1}$ and $\nu^{(\boldsymbol{\alpha})}_{\infty,1}$ does not contribute due to the reduction.
\subsection{Entries $(1,1)$}
 Let us turn our attention to the entry $(1,1)$, for this one needs the following result
\begin{lemma} On has for all $i\in \llbracket 1,g\rrbracket$
\begin{align}
\partial_{q_i}[Q_{r_\infty-4}]=&-1\,\,,\,\, \text{ if }r_\infty\geq 4\cr
\partial_{q_i}[Q_{m}]=&- \,q_i^{r_\infty-4-m}-\sum_{j=m+1}^{r_\infty-4} Q_{j} q_i^{j-1-m} \,\,,\,\, \forall \, m\in \llbracket 1,r_\infty-5\rrbracket
\end{align}
\end{lemma}
\begin{proof}
    The proof is a straightforward argument that is based on the fact that the change of coordinates and its derivative are rational functions with a pole located only at $\{ \infty \}$, taking the derivative 
    \begin{align}
        -\prod_{j \neq i} \left( \lambda -q_i \right) = \frac{-\prod_{j =1}^g \left( \lambda -q_i \right)}{\lambda-q_i} = \sum_{k=0}^{r_\infty -4} \partial_{q_i} (Q_k) \lambda^k
    \end{align}
    Expanding at the pole 
    \begin{align}
        \frac{-\prod_{j =1}^g \left( \lambda -q_i \right)}{\lambda-q_i} = & - \left( \sum_{k=0}^{r_\infty-4} Q_k \lambda^k + \lambda^{r_\infty-3} \right) \sum_{r=0}^\infty q_i^r \lambda^{-1-r} \nonumber \\
        =&-\sum_{k=0}^{r_\infty-4}\sum_{r=0}^{\infty} Q_{k} q_i^r \lambda^{k-1-r}-\sum_{r=0}^{r_\infty-4} q_i^r \lambda^{r_\infty-4-r} +O(\lambda^{-1})\cr
         \overset{m=k-1-r}{=}&-\sum_{m=0}^{r_\infty-5}\sum_{k=m+1}^{r_\infty-4} Q_{k} q_i^{k-1-m} \lambda^{m} -\sum_{m=0}^{r_\infty-4} q_i^{r_\infty-4-m} \lambda^{m} +O(\lambda^{-1})
    \end{align}
    By identifying the orders one gets the required result.
\end{proof}
One also needs another factor in the expansion of the matrices, the factor $Q(\lambda)$
\begin{lemma}One has the following expression 
\begin{align}
    Q(\lambda ) = & - \delta_{r_\infty \geq 4} P_{r_\infty-4} - \sum_{k=0}^{r_\infty-5} P_{k} \left( - \,\lambda_i^{r_\infty-4-m}-\sum_{j=m+1}^{r_\infty-4} Q_{j} \lambda_i^{j-1-m} \right) \cr
    =&  \sum_{k=0}^{r_\infty-4} P_{r_\infty-4 -k} \lambda^k -  \sum_{k=0}^{r_\infty-5}\sum_{m=0}^{r_\infty-5-k}  P_{m} Q_{k+1+m} \lambda^k
\end{align} \end{lemma}
\begin{proof}
    Let us first observe that $Q(\lambda)$ is a rational function of $\lambda$ with a pole at $\{ \infty \}$ of order $r_\infty-4$, and let us also note that $Q(q_i)=-p_i$, and thus, using the definition of the geometric coordinates one has
    \begin{align}
        Q(q_i) = & \delta_{r_\infty \geq 4} P_{r_\infty-4}  \partial_{q_i} Q_{r_\infty-4} + \sum_{k=0}^{r_\infty-5} P_{k} \partial_{q_i} Q_k \cr
        = & - \delta_{r_\infty \geq 4} P_{r_\infty-4} - \sum_{k=0}^{r_\infty-5} P_{k} \left( - \,q_i^{r_\infty-4-m}-\sum_{j=m+1}^{r_\infty-4} Q_{j} q_i^{j-1-m} \right) 
    \end{align}
One realizes the expression needed when evaluating at the local coordinate.
\end{proof}
This finally allows one to get the entry 
\begin{align}
    \td{L}_{1,1}(\lambda) = -Q(\lambda) - \left( \frac{1}{2} t_{2 r_\infty -2} \lambda + g_0 \right) \left(\sum_{k=0}^{r_\infty-4 } Q_k \lambda^k + \lambda^{r_\infty-3} \right)
\end{align}
Let us provide the explicit expression of the factor $g_0$, for this, one examines the local expansion of the connection matrix near the pole and gets $\td{L}_{1,1}(\lambda)= - \frac{1}{2} t_{2 r_\infty -2} \lambda^{r_\infty-2} - \frac{1}{2} t_{2 r_\infty -4} \lambda^{r_\infty-3}$, one has since $Q(\lambda = O(\lambda^{r_\infty-4})$ 
\begin{align}
    - g_0 - \frac{1}{2} t_{2 r_\infty -2} Q_{r_\infty-4} =- \frac{1}{2} t_{2 r_\infty -4}
\end{align}
This completes the expansion of the entry $(1,1)$. To obtain the auxiliary matrix's entry $(1,1)$ one needs a bit more information. One has in particular from the gauge transformation that
\beq \label{GaugeTdA11}[\td{A}_{\boldsymbol{\alpha}}(\lambda)]_{1,1}=[A_{\boldsymbol{\alpha}}(\lambda)]_{1,1}+\td{L}_{1,1}(\lambda)[A_{\boldsymbol{\alpha}}(\lambda)]_{1,2}\eeq 
Since we know the explicit expression of $[A_{\boldsymbol{\alpha}}(\lambda)]_{1,1}$ which is a rational function of the local coordinate with only a pole located at $\{ \infty \}$, one concludes that it contributes only with its constant term. Yet, one has the following result
\begin{lemma} One has in the geometric coordinates that
\begin{align}
     c_{\infty,0} = 0 
 \end{align}
\end{lemma}
\begin{proof}
Observe first that from the zero curvature equation one has
\beq \mathcal{L}_{\boldsymbol{\alpha}}[\td{L}_{1,2}(\lambda)]=\partial_{\lambda}\td{A}^{(\boldsymbol{\alpha})}_{1,2}(\lambda)-2\left(\td{L}_{1,1}(\lambda)\td{A}^{(\boldsymbol{\alpha})}_{1,2}(\lambda)-\td{L}_{1,2}(\lambda)\td{A}^{(\boldsymbol{\alpha})}_{1,1}(\lambda)\right)\eeq
where, we have eliminated the trace of the connection and auxiliary matrices since it hlds no contribution. From the gauge transformation one has $\td{A}^{(\boldsymbol{\alpha})}_{1,2}(\lambda)=\td{L}_{1,2}(\lambda)A^{(\boldsymbol{\alpha})}_{1,2}(\lambda)$ and $\td{A}^{(\boldsymbol{\alpha})}_{1,1}(\lambda)=A^{(\boldsymbol{\alpha})}_{1,1}(\lambda)+\td{L}_{1,1}(\lambda)A^{(\boldsymbol{\alpha})}_{1,2}(\lambda)$ so
\beq \mathcal{L}_{\boldsymbol{\alpha}}[\td{L}_{1,2}(\lambda)]=\partial_{\lambda}\td{A}^{(\boldsymbol{\alpha})}_{1,2}(\lambda)+2\td{L}_{1,2}(\lambda)A^{(\boldsymbol{\alpha})}_{1,1}(\lambda)
\eeq
 One has $A^{(\boldsymbol{\alpha})}_{1,1}(\lambda)=c_{\infty,0}+\frac{a_{1}}{\lambda}+O(\lambda^{-2})$, and $[\td{L}_{1,2}(\lambda)] = \lambda^{r_\infty-3} + Q_{r_\infty - 4 } \lambda^{r_\infty-4} + O(\lambda^{r_\infty-5})$. This gives the expected result,
 since one has from \autoref{geometricmatrices} that $\partial_{\lambda}\td{A}^{(\boldsymbol{\alpha})}_{1,2}(\lambda) = O(\lambda^{r_\infty-5})$
\end{proof}
This reduces the expression of the entry to only
\beq \label{GaugeTdA11}[\td{A}_{\boldsymbol{\alpha}}(\lambda)]_{1,1}= \td{L}_{1,1}(\lambda)[A_{\boldsymbol{\alpha}}(\lambda)]_{1,2}\eeq 
The remaining task is to evaluate this expansion at the pole, 
 \footnotesize{\bea
    && \left(  \sum_{i=1}^{r_\infty-2} \frac{\nu^{(\boldsymbol{\alpha})}_{\infty,i}}{\lambda^i} +O\left(\lambda^{-(r_\infty-1)}\right) \right) \bigg[ \sum_{k=0}^{r_\infty-4 } P_{r_\infty-4-k}\lambda^k + \sum_{k=0}^{r_\infty-5 } \sum_{m=0}^{r_\infty-5-k } P_m Q_{k+1+m} \lambda^k \nonumber \\
    &&-\left(\frac{1}{2}t_{\infty,2r_\infty-2}\lambda+g_0 \right)  \left(\sum_{k=0}^{r_\infty-4 } Q_k \lambda^k + \lambda^{r_\infty-3} \right) \bigg] \nonumber \\
    = && -\frac{1}{2} t_{\infty,2r_\infty-2} \sum_{i=1}^{r_\infty-2}  \nu^{(\boldsymbol{\alpha})}_{\infty,i} \lambda^{r_\infty-2 - i} + \sum_{i=1}^{r_\infty-4}  \nu^{(\boldsymbol{\alpha})}_{\infty,i} \left( P_{0} - \frac{1}{2} t_{\infty,2r_\infty-2}Q_{r_\infty-5} + g_0 Q_{r_\infty-4}   \right) \lambda^{r_\infty-4 - i }+ \nonumber \\
    & &\sum_{i=1}^{r_\infty-5} \nu^{(\boldsymbol{\alpha})}_{\infty,i} \sum_{k=0}^{r_\infty-5 } \left(  P_{r_\infty-4-k}\lambda^k + \sum_{m=0}^{r_\infty-5-k } P_m Q_{k+1+m} - Q_k g_0 -\frac{1}{2} t_{\infty,2r_\infty-2} Q_{k-1} \right) \lambda^{k-i} + O(\lambda^{-1}) \nonumber \\
    = && -\frac{1}{2} t_{\infty,2r_\infty-2} \sum_{j=0}^{r_\infty-3}  \nu^{(\boldsymbol{\alpha})}_{\infty,r_\infty-2-j} \lambda^{j} +\sum_{j=0}^{r_\infty-5}  \nu^{(\boldsymbol{\alpha})}_{\infty,r_\infty-4-j} \left( P_{0} - \frac{1}{2} t_{\infty,2r_\infty-2}Q_{r_\infty-5} + g_0 Q_{r_\infty-4}   \right) \lambda^{j }+ \nonumber \\
    && \sum_{j=0}^{r_\infty-6} \sum_{i=1}^{r_\infty-5-j} \nu^{(\boldsymbol{\alpha})}_{\infty,i} \left(  P_{r_\infty-4-i-j}\lambda^k + \sum_{m=0}^{r_\infty-5-k } P_m Q_{i+j+1+m} - Q_{i+j} g_0 -\frac{1}{2} t_{\infty,2r_\infty-2} Q_{i+j-1} \right) \lambda^j + O(\lambda^{-1}) \nonumber
 \eea }
 Which is exactly the needed expression.
 
\section{Proof of \autoref{Laxcoordinates}}\label{proofLaxcoordinates}
The only non-trivial computation needed is the one giving the entry $(1,1)$ of the auxiliary matrix, yet, the strategy is roughly the same as for the previous set of coordinates. One starts as always from the compatibility equation which gives
\beq \label{GaugeTdA11}[\td{A}_{\boldsymbol{\alpha}}(\lambda)]_{1,1}= \td{L}_{1,1}(\lambda)[A_{\boldsymbol{\alpha}}(\lambda)]_{1,2}\eeq 
One uses this to compute the entry in terms of the Lax coordinates, the explicit expansion is given by 
\footnotesize{
\bea   
 &&\left(  \sum_{i=1}^{r_\infty-2} \frac{\nu^{(\boldsymbol{\alpha})}_{\infty,i}}{\lambda^i} +O\left(\lambda^{-(r_\infty-1)}\right) \right) \bigg[ \sum_{k=0}^{r_\infty-4} R_k \lambda^k - \frac{1}{2} t_{\infty,2r_\infty-4} \lambda^{r_\infty-3} - \frac{1}{2}t_{\infty,2r_\infty-2}\lambda^{r_\infty-2}  \bigg] \cr
& &= \sum_{i=1}^{r_\infty-2}  \sum_{k=0}^{r_\infty-4} R_k \nu^{(\boldsymbol{\alpha})}_{\infty,i} \lambda^{k-i}  - \frac{1}{2} t_{\infty,2r_\infty-4} \sum_{i=1}^{r_\infty-3}  \nu^{(\boldsymbol{\alpha})}_{\infty,i}\lambda^{r_\infty - 3 - i} - \frac{1}{2}t_{\infty,2r_\infty-2}  \sum_{i=1}^{r_\infty-2}  \nu^{(\boldsymbol{\alpha})}_{\infty,i}\lambda^{r_\infty-2-i}  + O(\lambda^{-1}) \cr
&& = \sum_{j=0}^{r_\infty-5} \left( \sum_{i=1}^{r_\infty-4 - j }  R_{j+i} \nu^{(\boldsymbol{\alpha})}_{\infty,i} \right) \lambda^j  - \frac{1}{2} t_{\infty,2r_\infty-4} \sum_{j=0}^{r_\infty-4}  \nu^{(\boldsymbol{\alpha})}_{\infty,r_\infty-3-j}\lambda^{j} - \frac{1}{2}t_{\infty,2r_\infty-2}  \sum_{j=0}^{r_\infty-3}  \nu^{(\boldsymbol{\alpha})}_{\infty,r_\infty-2-j}\lambda^{j}+ O(\lambda^{-1}) \nonumber
\eea
}
\normalsize{which} is exactly what is required.

\section{Compatibility of the system} \label{AppD}
 To ensure the compatibility of the system, one has to ensure the commutativity of the differential operator $ [ \delta_{t_i}, \delta_{t_j} ]Q_k = 0$, to have this one writes the system after the reduction by 
\footnotesize{\beq \begin{pmatrix} \delta^{(\boldsymbol{\alpha_j})}_{\mathbf{\tau_j}}\left[\frac{Q_{\infty,r_\infty-5}}{r_\infty-4}\right]\\
\vdots\\
\delta^{(\boldsymbol{\alpha_j})}_{\mathbf{\tau_j}}\left[\frac{Q_{\infty,0}}{1}\right]\end{pmatrix}=
\begin{pmatrix}1&0&\dots&0\\
Q_{\infty,r_\infty-4}&1& &0\\
\vdots& \ddots &\ddots &\vdots \\
Q_{\infty,2}&\dots&Q_{\infty,r_\infty-4}&1
\end{pmatrix}\left( I + N_\tau \right)^{-1} \frac{1}{2r_\infty-2j-5} \mathbf{e}_{j}
\eeq}
Where the matrix $I+N_\tau$ is the matrix appearing in \eqref{Matrixentries} (with the exclusion of he last line) split into the identity and a nilpotent matrix. Its inverse is calculated via 
\begin{align}
    \left( I + N_\tau \right)^{-1} = I+ \sum_{k=1}^{r_\infty-6} (-1)^k N_\tau^k
\end{align}
 since the nilpotency index of $N_\tau$ is $g-2 = r_\infty-5$. Note that the commutator is equivalent to the following equation $\forall m \in \llbracket 0, r_\infty - 6 \rrbracket$
\begin{align} \label{Firstcondition3}
      \delta_{\mathbf{\tau_j}}^{(\boldsymbol{\alpha_j})} \left( \nu_{\infty,r_\infty-4-m}^{(\boldsymbol{\alpha}^{\tau_i})}+\sum_{k=m+2}^{r_\infty-4}\nu_{\infty,k-m-1}^{(\boldsymbol{\alpha}^{\tau_i})}Q_{\infty,k}\right) = \delta_{\mathbf{\tau_i}}^{(\boldsymbol{\alpha_i})} \left( \nu_{\infty,r_\infty-4-m}^{(\boldsymbol{\alpha}^{\tau_j})}+\sum_{k=m+2}^{r_\infty-4}\nu_{\infty,k-m-1}^{(\boldsymbol{\alpha}^{\tau_j})}Q_{\infty,k}\right)
    \end{align}
The first step is thus to get the explicit form of $\nu_{\infty,j}^{(\boldsymbol{\alpha}^{\tau_j})}$, for this one needs to invert the lower triangular matrix in \eqref{Matrixentries}, one has

\small{\beq \begin{pmatrix}1&0&\dots&\dots&\dots&0\\
0&1&0&\ddots&&0\\
\tau_1 &0&1&0&\ddots &\vdots\\
\vdots&\ddots &\ddots&\ddots&\ddots &\vdots\\
\vdots&\ddots &\ddots&\ddots&\ddots &\vdots\\
\tau_{g-3}&\tau_{g-4}&\dots& \tau_1&0&1 
\end{pmatrix}^{-1} = \begin{pmatrix}1&0&\dots&\dots&\dots&0\\
0&1&0&\ddots&&0\\
F_1(\tau_1) &0&1&0&\ddots &\vdots\\
\vdots&\ddots &\ddots&\ddots&\ddots &\vdots\\
\vdots&\ddots &\ddots&\ddots&\ddots &\vdots\\
F_{g-3}(\tau_1,\dots,\tau_{g-3})&F_{g-4}(\tau_1,\dots,\tau_{g-4})&\dots& F_1(\tau_1)&0&1 
\end{pmatrix}
\eeq}
\normalsize{with}
\beq \label{InvertLowerTriangularToeplitz} F_i(\tau_{1},\dots,\tau_i)=\sum_{\begin{subarray}{l}
(b_1,\dots,b_i)\in \mathbb{N}^i\\ \underset{j=1}{\overset{i}{\sum}}(j+1)b_j=i+1
\end{subarray}} \binom{\underset{j=1}{\overset{i}{\sum}}b_j}{b_1,\dots,b_i} (-1)^{\underset{j=1}{\overset{i}{\sum}}b_j} \,\tau_1^{b_1}\dots \tau_{i}^{b_i} \,\,,\,\, \forall \, i\geq 1
\eeq
For example, the first values of $\left(F_i(\tau_1,\dots,\tau_i)\right)_{1\leq i\leq 5}$ are
\bea F_1(\tau_1)&=&-\tau_1\cr
F_2(\tau_1,\tau_2)&=&-\tau_2\cr
F_3(\tau_1,\tau_2,\tau_3)&=&\tau_1^2-\tau_3\cr
F_4(\tau_1,\tau_2,\tau_3,\tau_4)&=&2\tau_1\tau_2-\tau_4\cr
F_5(\tau_1,\dots,\tau_5)&=&-\tau_1^3+2\tau_1\tau_3+\tau_2^2-\tau_5
\eea
Due to the term $\mathbf{e}_j$ on the r.h.s. one obtains
\beq \label{Reducednu}\nu^{(\boldsymbol{\alpha}^{\tau_j})}_{\infty,k}=\frac{2}{2r_\infty-2j-5}\left(\delta_{j,k}+F_{j-k-1}(\tau_{1},\dots,\tau_{j-k-1}) \delta_{k\leq j-2}\right) \,\,,\,\,\forall\, (j,k)\in \llbracket 1,g\rrbracket^2\eeq
Despite the fact that the inverse of these coefficients is given, the manipulation of these expressions is rather involved. For this, let us first write the system \eqref{system} in the following way 
\begin{align}
    \delta_{t} \mathcal{Q} = Q_0 T^{-1} \alpha
\end{align}
Specializing to each component of the general exterior derivative operator and using the reduction, one has this system equivalent to 
    \footnotesize{\bea
    &\begin{pmatrix}2&0&\dots&& &0\\
0& 2&0& \ddots& &\vdots\\
t_{\infty,2r_\infty-7}&\ddots&\ddots& \ddots&&\vdots\\
\vdots& \ddots& \ddots&&&\vdots \\
t_{\infty,7}& t_{\infty,9} & \ddots &\ddots& & 2\end{pmatrix} \begin{pmatrix}
    \frac{1}{r_\infty-4} & 0 & \dots & 0 \\ 
    0 & \frac{1}{r_\infty-5} & \dots & 0 \\
    0 & \ddots & \ddots & 0 \\
    0 & \dots & \dots &\frac{1}{1}
    \end{pmatrix} \cr
&\begin{pmatrix} \delta_{t_{\infty,2r_\infty-7}}[Q_{\infty,r_\infty-5}]& \dots& \delta_{t_{\infty,3}}[Q_{\infty,r_\infty-5}]\\
\vdots\\
\delta_{t_{\infty,2r_\infty-7}}[Q_{\infty,0}]& \dots &\delta_{t_{\infty,3}}[Q_{\infty,0}]\end{pmatrix}\cr 
&= \begin{pmatrix}1&0&\dots&0\\
Q_{\infty,r_\infty-4}&1& &0\\
\vdots& \ddots &\ddots &\vdots \\
Q_{\infty,2}&\dots&Q_{\infty,r_\infty-4}&1
\end{pmatrix}
\begin{pmatrix}\frac{2}{2r_\infty-7}&0&\dots&0\\
0&\frac{2}{2r_\infty-9}& &0\\
\vdots& \ddots &\ddots &\vdots \\
0&\dots&\dots &\frac{2}{3}
\end{pmatrix}
\eea}
As discussed in the theorem, the solution of the above system is recursive due to the triangular Toeplitz matrices. For completeness, let us present the second line to show this 
\begin{align}
    \frac{1}{r_\infty-5} \delta_{t_{\infty,k}} [Q_{\infty, r_\infty-6}] = \delta_{k, 2r_\infty-7} \frac{1}{2r_\infty-7}  Q_{\infty,r_\infty-4} + \delta_{k, 2r_\infty-9} \frac{1}{2r_\infty-9} 
\end{align}
with a solution given by 
\begin{align}
    Q_{\infty, r_\infty-6} = \frac{ r_\infty-5}{2r_\infty-7} u_{r_\infty-4} t_{\infty
    , 2r_\infty-7} + \frac{ r_\infty-5}{2r_\infty-9 }  t_{\infty
    , 2r_\infty-9}      +  u_{r_\infty-6}
\end{align}
One immediately realizes the recursion of the solution. To complete the proof, one must check the compatibility of the system, this is done through the commutation of the differential operators, to this end, one writes the l.h.s. of the system as  
\bea
  &\begin{pmatrix}2&0&\dots&& &0\\
0& 2&0& \ddots& &\vdots\\
t_{\infty,2r_\infty-7}&\ddots&\ddots& \ddots&&\vdots\\
\vdots& \ddots& \ddots&&&\vdots \\
t_{\infty,7}& t_{\infty,9} & \ddots &\ddots& & 2\end{pmatrix} \begin{pmatrix}
    \frac{1}{r_\infty-4} & 0 & \dots & 0 \\ 
    0 & \frac{1}{r_\infty-5} & \dots & 0 \\
    0 & \ddots & \ddots & 0 \\
    0 & \dots & \dots &\frac{1}{1}
    \end{pmatrix} \cr
&\begin{pmatrix} \delta_{t_{\infty,2r_\infty-7}}[Q_{\infty,r_\infty-5}]& 0 & \dots& 0 \\
\delta_{t_{\infty,2r_\infty-7}}[Q_{\infty,r_\infty-6}] & \delta_{t_{\infty,2r_\infty-9}}[Q_{\infty,r_\infty-6}] & \dots  & 0
\\
\vdots& \ddots & \dots & 0\\
\delta_{t_{\infty,2r_\infty-7}}[Q_{\infty,0}]& \dots &\dots   &\delta_{t_{\infty,3}}[Q_{\infty,0}]\end{pmatrix}
\eea
The proof of the compatibility of the system follows using a recursive analysis, in particular, note that one has trivially $\delta_{t_i} \delta_{t_j} Q_{\infty,r_\infty-4}= \delta_{t_j} \delta_{t_i} Q_{\infty,r_\infty-4} =0$. Now suppose that for all $Q_{\infty,j+1}$ one has the following 
\begin{align}
    Q_{\infty,j+1} = \sum_i g_i({t_{\infty,i}})
\end{align}
where the coefficients $g_i$ are dependent on only one time. Note that proving that the coordinate ${Q_{\infty,j}}$ has this same structure is sufficient. let us prove that this forces this condition for $Q_{\infty,j}$. The above system has a $k-$th line given by 
\begin{align} \label{BiglineK}
    \frac{1}{2r_\infty-5-2k} =& \frac{1}{r_\infty-3-k} \delta_{t_{\infty,2r_\infty-5-2k}} [Q_{\infty,r_\infty-4-k}] \nonumber \\
    \frac{ Q_{\infty,r_\infty-4}}{2r_\infty-3-2k} = & \frac{1}{r_\infty-3-k} \delta_{t_{\infty,2r_\infty-3-2k}} [Q_{\infty,r_\infty-4-k}] \nonumber \\
    \frac{ Q_{\infty,r_\infty-5}}{2r_\infty-1-2k}= & \frac{1}{r_\infty-1-k} t_{\infty,2 r_\infty-7} \delta_{t_{\infty,2r_\infty-1-2k}} [Q_{\infty,r_\infty-2-k}] + \frac{1}{r_\infty-3-k} \delta_{t_{\infty,2r_\infty-1-2k}} [Q_{\infty,r_\infty-4-k}] \nonumber \\
    \vdots \nonumber \\
    \frac{ Q_{\infty,r_\infty-2-k} }{2 r_\infty-7} = & \frac{1}{r_\infty-4}  t_{\infty,2r_\infty-1-2k} \delta_{t_{\infty,2r_\infty-7}}[Q_{\infty,r_\infty-5}]+ \frac{1}{r_\infty-5} t_{\infty,2r_\infty-3-2k} \delta_{t_{\infty,2r_\infty-7}}[Q_{\infty,r_\infty-6}] + \dots \nonumber \\
    & +0 + \frac{1}{r_\infty-3-k} \delta_{t_{\infty,2r_\infty-7}} [Q_{\infty,r_\infty-4-k}]
\end{align}
Note that the proof for the first two entries of the $k-$th line is trivially true, however, to prove this for the other entries, one needs to use an additional fact: the coordinates have the following time dependence
\begin{align}
    Q_{\infty,r_\infty-4} =  & 0 \nonumber \\
     Q_{\infty,r_\infty-5} =  & f(t_{2r_\infty-7}) \\
     Q_{\infty,r_\infty-6} =  & f(t_{2r_\infty-7}, t_{\infty,2r_\infty-9}) \\
     \vdots \nonumber \\
     Q_{\infty,0} =  & f(t_{2 r_\infty-7}, t_{\infty,2 r_\infty-9}, \dots, t_{\infty,3})
\end{align}
Keeping this in mind, let us consider the third line of \eqref{BiglineK} and differentiate it with $\delta_{t_{\infty,j}}$, due to the supposition of the recursion, one gets a non-trivial case only when $j=2r_\infty-7$. Thus, one has
\begin{align}
\frac{1}{r_\infty-3-k} \delta_{t_{\infty,2r_\infty-7}} \delta_{t_{\infty,2r_\infty-1-2k}} [Q_{\infty,r_\infty-4-k}] = &\frac{1}{2r_\infty-1-2k} \delta_{t_{\infty,2r_\infty-7}}Q_{\infty,r_\infty-5} \nonumber \\
        & - \frac{1}{r_\infty-1-k} \delta_{t_{\infty,2r_\infty-1-2k}} [Q_{\infty,r_\infty-2-k}]  
\end{align}
Note that the r.h.s. of the above expression is a scalar due to the supposition of the recursion and thus satisfies our assumption. O,e could apply the same analogy to any line presented in \eqref{BiglineK} and get the same conclusion thus proving the commutativity of the exterior differential operators and thus the compatibility. Finally let us denote that the general solution of the system admits the following form
\begin{align}
    \begin{pmatrix}
        Q_{\infty,r_\infty-4} \\
        Q_{\infty,r_\infty-5} \\
        \vdots \\
        Q_{\infty,0}
    \end{pmatrix} = \begin{pmatrix}
    0 &0 &0 & \dots &  0\\
    f_{2,1} (t_{\infty,2r_\infty-7}) & 0 & 0 & \dots & 0 \\
    f_{3,1} (t_{\infty,2r_\infty-7},t_{\infty,2r_\infty-9}) & f_{3,2} (t_{\infty,2r_\infty-7},t_{\infty,2r_\infty-9}) & 0 & \dots & 0 \\
    \vdots & \ddots & & \dots & 0 \\
    f_{r_\infty-3,1} (t_{\infty,2r_\infty-7}, \dots, t_{\infty,3}) & \dots & \dots &  & 0
    \end{pmatrix} 
    \begin{pmatrix}
        u_{\infty,r_\infty-4}  \\
        u_{\infty,r_\infty-5} \\
        \vdots \\
        u_{\infty,0} 
    \end{pmatrix}
\end{align}

\addcontentsline{toc}{section}{References}
\bibliographystyle{plain}
\bibliography{Biblio}

\end{document}